\documentclass[pdflatex,iicol,sn-mathphys-num]{sn-jnl}

\usepackage{graphicx}%
\usepackage{multirow}%
\usepackage{amsmath,amssymb,amsfonts}%
\usepackage{amsthm}%
\usepackage{mathrsfs}%
\usepackage[title]{appendix}%
\usepackage{xcolor}%
\usepackage{textcomp}%
\usepackage{manyfoot}%
\usepackage{booktabs}%
\usepackage{algorithm}%
\usepackage{algorithmicx}%
\usepackage{algpseudocode}%
\usepackage{adjustbox}
\usepackage{listings}%

\let\orcidlogo\relax
\usepackage{orcidlink}

\theoremstyle{thmstyleone}%
\newtheorem{theorem}{Theorem}
\newtheorem{proposition}[theorem]{Proposition}%

\theoremstyle{thmstyletwo}%

\theoremstyle{thmstylethree}%
\newtheorem{definition}{Definition}%

\usepackage{parskip}
\usepackage{braket}
\usepackage{makecell}
\newtheorem{corollary}[theorem]{Corollary}
\newtheorem{lemma}[theorem]{Lemma}
\newenvironment{restate}[2]{%
  \par\noindent\textbf{#1 \ref{#2}. }\itshape%
}{\par\medskip}

\iffalse
\newcommand{\comments}[1]{\footnote{\textcolor{blue}{\textit{#1}}}}

\newcommand{\LW}[1]{\textcolor{blue}{#1}}

\else
\newcommand{\comments}[1]{}

\newcommand{\LW}[1]{#1}

\fi

\begin{document}

\title[Article Title]{Learning PDEs for Portfolio Optimization with Quantum Physics-Informed Neural Networks}



\author*[1]{\fnm{Letao} \sur{Wang}}\email{letao.wang@centralesupelec.fr}

\author[1,2]{\fnm{Abdel} \sur{Lisser}}\email{abdel.lisser@centralesupelec.fr}

\author[1]{\fnm{Sreejith} \sur{Sreekumar}}\email{sreejith.sreekumar@centralesupelec.fr}

\author[1]{\fnm{Zeno} \sur{Toffano}}\email{zeno.toffano@centralesupelec.fr}

\affil[1]{\orgdiv{Laboratory of Signals and Systems},
\orgname{CentraleSupélec, CNRS, Paris-Saclay University},
\orgaddress{\street{3 Rue Joliot Curie}, \city{Gif-sur-Yvette},
\postcode{91190}, \state{Île-de-France}, \country{France}}}

\affil[2]{\orgdiv{Fédération de Mathématiques de CentraleSupélec},
\orgname{CentraleSupélec, CNRS, Paris-Saclay University},
\orgaddress{\street{3 Rue Joliot Curie}, \city{Gif-sur-Yvette},
\postcode{91190}, \state{Île-de-France}, \country{France}}}


\abstract{Partial differential equations (PDEs) play a crucial role in financial mathematics, particularly in portfolio optimization, and solving them using classical numerical or neural network methods has always posed significant challenges. Here,  we investigate the potential role of quantum circuits for solving PDEs. We design a parameterized quantum circuit (PQC) for implementing a polynomial based on tensor rank decomposition, reducing the quantum resource complexity from exponential to polynomial when the corresponding tensor rank is moderate. Building on this circuit, we develop a Quantum Physics-Informed Neural Network (QPINN) and a Quantum-inspired PINN, both of which guarantee the existence of an approximation of the PDE solution, and this approximation can be represented as a polynomial that incorporates tensor rank decomposition. \LW{Numerical experiments are conducted on the Hamilton--Jacobi--Bellman (HJB) PDE arising from the Merton portfolio optimization problem, which determines the optimal investment fraction between a risky and a risk-free asset. The results show that our quantum models achieve lower losses and approximation errors than a classical fully connected PINN while using substantially fewer trainable parameters.} Our quantum models further outperform a classical PINN constructed to share a \LW{similar} inductive bias, providing experimental evidence of quantum-induced improvement \LW{in the tested settings} and highlighting a resource-efficient pathway toward classical and near-term quantum solvers for PDEs \LW{with exploitable solution structure}. }





\maketitle

\section{Introduction}
\bmhead{Classical challenges}
In financial mathematics, Partial Differential Equations (PDEs) are fundamental to modeling and valuation of financial derivatives, risk management, and optimal portfolio strategies. A cornerstone in this domain is the Merton portfolio optimization problem introduced by \citet{Merton_1969}, which aims to determine an optimal investment strategy that maximizes the expected utility of the terminal wealth of an investor under stochastic market dynamics. Directly finding the optimal strategy is often difficult due to the randomness of dynamics, so it is commonly transformed into a Hamilton–Jacobi–Bellman (HJB) PDE, which characterizes both the maximum expected utility and the corresponding optimal investment strategy. Therefore, the Merton portfolio optimization problem can be solved by finding the solution to its corresponding HJB PDE \citep{HJB_2025,jump_diff,XuhuiWang_2022,HJB_2014,Gollier_book}.

Classical numerical PDE solvers, such as the finite difference method (FDM) and the finite element method (FEM), can be used to approximate PDE solutions. However, achieving high accuracy with these traditional deterministic methods often requires a significant computational cost \citep{Luo2025PhysicsInformed}. To address this limitation, non-deterministic algorithms have also been explored. For example, Physics-Informed Neural Networks (PINNs) have emerged as a promising mesh-free alternative that uses deep learning to solve PDEs by embedding physical laws into the loss function \citep{raissi2019physics}. Despite their promise, PINNs also face several limitations, such as slow convergence, difficulty capturing high-frequency features, and scalability issues in high-dimensional problems \citep{Luo2025PhysicsInformed}. These limitations of classical PDE solvers motivate the exploration of a new computational paradigm as a possible route to solve PDEs more efficiently.

\bmhead{Quantum opportunities}
Quantum computing, with its ability to encode and process information in exponentially large-dimensional Hilbert spaces, offers a new frontier for accelerating scientific computing tasks. It has been considered as a new computational model with the potential to solve problems beyond the reach of classical computers. Although the construction of practical quantum hardware remains challenging, theoretical results have shown that quantum algorithms can achieve significant complexity advantages over classical methods, for example in factoring \citep{Shor_1997} and solving linear systems \citep{Harrow_2009} to address linear ordinary differential equations \citep{Berry_2017}. Such results encourage people to investigate the potential of quantum computers for solving PDEs, which can be broadly categorized into two types \citep{Hunout_2025}: fully quantum algorithms that directly encode differential equations within quantum circuits and hybrid quantum-classical algorithms that integrate quantum components into classical optimization loops.

The first fully quantum algorithms for differential equations were encoded by vector amplitudes built on the quantum amplitude amplification algorithm (QAAA) and the quantum amplitude estimation algorithm (QAEA)   \citep{Brassard_2002}. These methods typically applied differentiation schemes over discretized variable spaces. Subsequently, \citet{kacewicz2006} proposed the Kacewicz quantum algorithm to solve initial value problems (IVP), i.e., ordinary differential equations (ODEs) with initial conditions using the QAAA and QAEA algorithms \citep{Brassard_2002}. Moreover, certain works extended the Kacewicz quantum algorithm to differential equations, including the Navier–Stokes equations \citep{Gaitan2020}, the Burgers equation \citep{Oz2022Solving}, and the heat equation \citep{OzSanKara2023}.

However, fully quantum algorithms for solving PDEs typically require fault-tolerant quantum hardware and efficient quantum random access memory (qRAM) \citep{Giovannetti_2008}, which remain infeasible in the near term \citep{Matteo_2020}. This has motivated the development of hybrid quantum–classical approaches \citep{Preskill_2018}. In parallel, quantum machine learning (QML) has emerged as a broad research field between quantum computing and machine learning \citep{Biamonte2017QuantumMachineLearning}. A popular approach to QML is based on parameterized quantum circuits (PQCs) \citep{Benedetti_2019}, whose training can be formulated as variational quantum algorithms (VQAs) \citep{Cerezo_2021} within the hybrid approach, making them suitable for near-term quantum devices that lack full fault tolerance or error correction capabilities. For instance, such hybrid approaches have been proposed for solving molecular ground-state energy estimation \citep{Peruzzo2014VQE}, quantum classification \citep{Schuld-2020}, and entropy estimation \citep{GSM-2024,SGM-2026}, to name a few applications.

Building on the idea of VQAs, \citet{Kyriienko_2021} proposed differentiable quantum circuits (DQCs) as a framework for solving PDEs, where a PQC is trained so that its expected measurement value serves as the output of the model, thus approximating the solution while minimizing the loss and enforcing the target PDE. Inspired by DQCs and the success of classical PINNs, several groups extended the DQC framework into Quantum Physics-Informed Neural Networks (QPINNs) \citep{Markidis_2022,Trahan_2024,Hegde2024Quantum,Berger2025Trainable,panichi2025quantumphysicsinformedneural,Farea_2025,Chen2026,Lantigua2026}, which can be seen as the quantum analog of PINNs. QPINNs can be implemented on both continuous-variable \citep{Markidis_2022,panichi2025quantumphysicsinformedneural} and discrete-variable (qubit-based) \citep{Berger2025Trainable,Trahan_2024} architectures, and the discrete-variable versions are sometimes referred to as VQAs for PDEs. In this work, we focus on the discrete-variable architecture.

\bmhead{Quantum limitations}
Precisely, several QPINN studies \citep{Hegde2024Quantum,Berger2025Trainable} employ a quantum model known as the quantum Fourier model (QFM) or re-uploading model, originally introduced by \citet{Schuld2021Effect}. In this approach, the circuit is made up of alternating encoding and trainable layers. The function induced by the quantum model, also called the \textit{hypothesis function} (the outputs of the hypothesis function are identical to those of the quantum model), can be expressed as a Fourier series. By varying all circuit parameters, these hypothesis functions collectively form the \textit{hypothesis space} of the model. Other QPINN models can also employ hypothesis functions based on Chebyshev polynomials \citep{Kyriienko_2021} or Lagrange polynomials \citep{Hunout_2025}. Moreover, QPINNs with QFM \citep{Hegde2024Quantum, Berger2025Trainable} and other VQAs designed to solve PDEs \citep{Kyriienko_2021,Hunout_2025} often employ a hardware efficient ansatz (HEA) \citep{Kandala_2017} as the training block. In practice, this introduces hidden constraints and makes it difficult to precisely characterize the hypothesis space from a theoretical perspective. Such limitations highlight a broader issue in VQA: most quantum models offer few theoretical guarantees that their hypothesis space contains an appropriate approximation of the target function \citep{chang2025primer}.

To better understand and characterize the hypothesis space, we turn to the quantum signal processing (QSP) framework \citep{Gily_n_2019}, where hypothesis functions can be expressed as univariate polynomials. Subsequently, \citet{yu2024nonasymptotic} extended QSP to multivariate settings, and this analysis shows that the quantum circuit resource complexity (depth, number of parameters) grows exponentially with the degree of the polynomial, making the approach impractical. 

\bmhead{Our approaches}
To address these limitations, we introduce an efficient QPINN model whose hypothesis space provably contains a \textit{tensor-decomposed} polynomial that accurately approximates the target function, where tensor-decomposed refers to representations that can be expressed as sums of products of univariate components, with the number of summation terms corresponding to the \textit{tensor rank} \citep{kolda2009tensor}. \LW{The tensor-decomposed structure applies only to some cases rather than to PDE solutions in general \citep{Logan2015}, and therefore serves as an inductive bias only when such structure is present. For PDE solutions that are not well aligned with this structure, we later introduce an added entangling layer in the QPINN to increase expressivity in a problem-oriented way and provide an additional inductive bias.} By further approximating each univariate function with a univariate polynomial, our proposed model effectively exploits this structural property, thereby reducing the quantum resource complexity for representing PDE solutions, from exponential to polynomial complexity when the tensor rank is not large. In addition to the quantum model, we also propose a quantum-inspired PINN obtained by a slight modification of the QPINN architecture, whose hypothesis functions can be tensor-decomposed polynomials and can be simulated efficiently on classical computers.

As we discussed before, the hypothesis functions of previous QPINNs can often be expressed through Fourier series (or other polynomial families) together with HEAs, but the actual hypothesis space is not the full set of Fourier series of a fixed degree with arbitrary coefficients, making the practical hypothesis space smaller than the ideal functional family. This effectively provides an upper bound on the hypothesis space size. In contrast, we consider the opposite perspective. The hypothesis space of our QPINN contains all tensor-decomposed polynomials of a fixed degree and bounded tensor rank, meaning that all possible coefficient choices are representable. As a result, the minimal distance (i.e., approximation error) between the target function and the hypothesis space can be explicitly estimated when the target function satisfies the required regularity conditions (e.g., continuity). Although the actual hypothesis space can be much larger than the set of tensor-decomposed polynomials, this set provides a lower bound on the hypothesis space size, thereby guaranteeing the existence of an approximation to the PDE solution within the QPINN hypothesis space.

\subsection*{Related works}
Similar ideas to our QPINN and quantum-inspired PINN have also been explored on the classical side. Polynomial hypothesis functions have been studied in PINNs \citep{Tang2023PIPINN}, while tensor-decomposition techniques have been incorporated into classical neural networks \citep{cohen2016,Chrysos_2021}. Our approaches share conceptual similarities with the framework proposed by \citet{Vemuri2025FunctionalTensor}, where the classical PINN learns each separable component of the tensor-decomposed hypothesis functions in order to efficiently approximate PDE solutions.

Some studies have demonstrated that orthogonal polynomials, such as Chebyshev polynomials, can offer theoretical advantages over standard non-orthogonal polynomials for solving PDEs \citep{powell1981approximation,boyd2001chebyshev}. One quantum approach has also employed Chebyshev polynomials as hypothesis functions for solving PDEs \citep{Kyriienko_2021}. 

However, it was experimentally observed in \citet{Tang2023PIPINN} that standard non-orthogonal polynomials achieve significantly faster training convergence than Chebyshev polynomials, while exhibiting a comparable error after convergence when solving low-dimensional PDEs such as the Burgers equation, Sine-Gordon, and Allen-Cahn equations. Standard non-orthogonal, Chebyshev, and tensor-decomposed polynomials represent particular subclasses within the general family of polynomial hypothesis functions, and it is still not well understood which hypothesis function representation is most appropriate for solving particular classes of PDEs. This observation motivates our investigation into quantum models whose hypothesis spaces are designed to contain polynomial hypothesis functions, as follows:

\subsection*{Main contribution}
\begin{itemize}
    \item We provide a detailed analysis of the quantum resource cost for implementing univariate, multivariate, and tensor-decomposed polynomials, and compare our results with existing approaches. While the implementation of a general multivariate polynomial requires exponentially large cost complexity, our tensor-decomposed model reduces the cost complexity to polynomial when the corresponding tensor rank is not very large, making the implementation on near-term quantum hardware more feasible.
    \item We incorporate the tensor-decomposed model to propose a QPINN and a Quantum-inspired PINN. Both models guarantee the existence of a tensor-decomposed polynomial approximating the PDE solution, with the approximation error depending on the PDE and the model parameters. Owing to entanglement, the QPINN has a richer hypothesis space and therefore greater expressivity.
    \item \LW{We conduct numerical simulations on two HJB PDEs in portfolio optimization, comparing our QPINN and Quantum-inspired PINN with classical PINNs. The experiments include a two-dimensional input case with a separable analytical solution and a three-dimensional input case with a non-separable analytical solution. In particular, we include a classical counterpart PINN that shares a similar tensor-decomposed inductive bias with our quantum models, as well as a fully connected PINN with significantly more parameters, and demonstrate that our QPINN and Quantum-inspired PINN outperform classical PINNs in the reported Merton HJB experiments.}
\end{itemize}

This paper is organized as follows. We first introduce the context of the Merton portfolio optimization problem in Section \ref{sec: merton portfolio optimization}. The general framework of PINNs and QPINNs is presented in Section \ref{sec:QPINN}, and the polynomials considered in our quantum models are described in Section \ref{sec:polynomials summary}. Section \ref{sec: summary} summarizes the main results on the quantum resource complexity required to implement general and tensor-decomposed polynomials, while Sections \ref{sec:poly} and \ref{sec: TD poly} provide the detailed descriptions. We then present our QPINN and Quantum-inspired PINN in Section \ref{sec:our_QPINN}. Finally, Section \ref{sec:experiment} reports the QPINN and PINN experimental simulations, and Section \ref{sec:conclusion} concludes the paper.

\section{Background}

In this section, we present the Merton portfolio optimization problem, how it can be solved within the QPINN framework, and define the corresponding polynomials needed.

\subsection{Merton Portfolio Optimization}
\label{sec: merton portfolio optimization}
The Merton portfolio optimization problem, introduced by \citet{Merton_1969}, studies how an investor should allocate his wealth with a choice between a risk-free asset (e.g., money deposited in a bank account) with a fixed interest rate $r(t)$ and a risky asset (e.g., a stock) with expected return $\mu(t)$ and volatility $\sigma(t)$. At each time $t$, the investor chooses a fraction $\alpha(t)$ of his wealth $X(t)$ to invest in the risky asset, while the remaining fraction $1-\alpha(t)$ is placed in the risk-free asset. The objective is to maximize the expected utility of the terminal wealth, which corresponds to finding the function
\begin{equation}
v(t, x)=\sup_{\alpha(t)} \mathbb{E}\Big[ U\big(X(T)\big)\,\big|\, X(t)=x \Big],
\end{equation}
where $U(x)$ is the investor's utility function that encodes risk preferences, and $v(t,x)$ is the maximum expected utility that can be achieved by following the optimal strategy from the current time $t$, with wealth $x$, up to the terminal time $T$. In the case of the Merton portfolio optimization model without jumps and with constant parameters, once the value function $v(t, x)$ is known, the optimal investment fraction $\hat{\alpha}(t)$ can be obtained directly, allowing the investor to achieve the best possible growth of wealth.

The dynamics of $X(t)$ follows a stochastic differential equation (SDE), which is often challenging to solve directly. Hence, one typically converts the problem into a PDE problem. In general settings, analytical solutions are unavailable, and we instead obtain the solution to the Merton portfolio problem by numerically solving the corresponding PDE. But in our setting, we can make the following assumptions: 
\begin{itemize}
    \item The expected return $\mu$, volatility $\sigma$, and fixed interest rate $r$ are constants, with $\mu > r$.
    \item The investor's utility function is $U(x)=\frac{x^\gamma}{\gamma}, x>0$ with a risk parameter $\gamma \in (0,1)$.
\end{itemize}
The investor's attitude towards risk determines the curvature of the utility function $U(x)$, and $\gamma \in (0,1)$ indicates risk aversion and implies a concave utility function \citep{Gollier_book}. Under these assumptions, the function $v(t,x)$ satisfies the Hamilton--Jacobi--Bellman (HJB) PDE on the domain $(t, x) \in[0, T] \times[0,+\infty)$:
\begin{equation}
\label{eq:HJB_PDE}
\begin{aligned}
\partial_t v(t, x)&\partial^2_x v(t, x)+\partial_x v(t, x) \partial^2_x v(t, x) r x\\
-\frac{1}{2} &\left(\frac{\mu-r}{\sigma}\right)^2 \partial_x v(t, x)^2=0,\\[1pt]
v(T,x)&=\frac{x^\gamma}{\gamma},\\
v(t,1)&=\exp (-k(T-t)) \frac{1}{\gamma}, 
\end{aligned}
\end{equation}
where $k=\frac{1}{2} \frac{\gamma}{\gamma-1}\left(\frac{\mu-r}{\sigma}\right)^2-r \gamma$. The analytical solution is given by
\begin{equation}
\label{eq:analytical_solution}
v(t, x)=\exp (-k(T-t)) \frac{x^\gamma}{\gamma},
\end{equation}
and the optimal investment fraction is given by $\hat{\alpha}=\frac{1}{1-\gamma} \frac{\mu-r}{\sigma^2}$. Note that the boundary conditions $v(T,x),v(t,1)$ in Equation~\eqref{eq:HJB_PDE} are not unique; other valid boundary conditions may also be chosen. See Appendix \ref{sec:appendix Merton portfolio optimization} for detailed derivation and explanation. 

\subsubsection{Volatility as an input variable}
\LW{In Equation~\eqref{eq:HJB_PDE}, the volatility $\sigma$ is fixed, because the equation represents a single Merton portfolio problem with a fixed market model. In practice, one often needs to evaluate the model under different volatility values. Hence, we also consider a version in which the volatility $\sigma$ is treated as an input variable and ranges over an interval. This increases the input dimension from $2$ to $3$, changing the model input from $(t,x)$ to $(t,x,\sigma)$. It represents a family of HJB PDEs with different fixed volatility values and allows us to examine how the solution changes as $\sigma$ changes. This setting can be viewed as a parametric PDE, where $\sigma$ acts as a parameter and the model learns the corresponding solution map over the parameter range.}

\LW{For each $\sigma\in[\sigma_{\min},\sigma_{\max}]$ and $(t,x)\in[0,T]\times[0,+\infty)$, we write this HJB PDE as
\begin{equation}
\label{eq:parametric_volatility_HJB}
\begin{aligned}
\partial_t v(t, x, \sigma)&\partial^2_x v(t, x, \sigma)+\partial_x v(t, x, \sigma) \partial^2_x v(t, x, \sigma) r x\\
-\frac{1}{2} &\left(\frac{\mu-r}{\sigma}\right)^2 \partial_x v(t, x, \sigma)^2=0,\\[1pt]
v(T,x,\sigma)&=\frac{x^\gamma}{\gamma},\\
v(t,1,\sigma)&=\exp (-k(\sigma)(T-t)) \frac{1}{\gamma}.
\end{aligned}
\end{equation}
where $k(\sigma)=
\frac{1}{2}\frac{\gamma}{\gamma-1}
\left(\frac{\mu-r}{\sigma}\right)^2
-r\gamma$. The corresponding analytical solution is
\begin{equation}
\label{eq:parametric_volatility_solution}
v(t, x, \sigma)=
\exp (-k(\sigma)(T-t))\frac{x^\gamma}{\gamma},
\end{equation}
The corresponding optimal investment fraction is given by $\hat{\alpha}(\sigma)=\frac{1}{1-\gamma} \frac{\mu-r}{\sigma^2}$. For each fixed value of $\sigma$, Equation~\eqref{eq:parametric_volatility_HJB} defines one HJB PDE, with the solution given by Equation~\eqref{eq:parametric_volatility_solution}. The notation $v(t, x, \sigma)$ simply collects these independent solutions $v(t, x)$ in Equation~\eqref{eq:analytical_solution} as $\sigma$ varies. Instead of solving Equation~\eqref{eq:HJB_PDE} separately for many values of $\sigma$, we train one model for the map $(t,x,\sigma)\mapsto v(t, x, \sigma)$ to approximate this family over $[\sigma_{\min},\sigma_{\max}]$. This makes it convenient to evaluate the solution and the corresponding optimal investment fraction for different volatility values, which is useful in practical portfolio problems.}

The portfolio optimization problem can alternatively adopt machine learning based approaches that directly learn the optimal fraction through empirical utility maximization, without deriving the HJB PDE as shown by \citet{Kopeliovich_2024}. In this work, we focus on the HJB PDE method. The following section introduces the QPINN framework for addressing general PDEs, with a specific application to the HJB PDE.

\subsection{Quantum/Classical PINN framework}
\label{sec:QPINN}
A general PDE can typically be expressed in the following form:
\begin{equation}
\label{eq:PDE}
\begin{aligned}
D_{x} (u;\lambda ) & =h(x),\ \ \forall x\in {\Omega}\subset \mathbb{R}^{d}\\
B_{k} (u) & =g_{k} (x),\ \ \forall x\in \partial {\Omega}\subset \mathbb{R}^{d} ,\ \ k=1,2,\dotsc ,n_{b}
\end{aligned}
\end{equation}
where 
\begin{itemize}
\item ${\Omega}$ is a domain with boundary $\partial {\Omega}$.
\item $u:{\Omega}\cup \partial \Omega \rightarrow \mathbb{R}$ is the unknown solution.
\item $D_{x} (\cdot ;\lambda )$ is a (possibly parameterized) differential operator with parameters $\lambda $.
\item $h:{\Omega}\rightarrow \mathbb{R}$ is a given source function.
\item $B_{k} (\cdot )$ is the $k^{\text{th }}$ boundary condition.
\item $g_{k} :\partial {\Omega}\rightarrow \mathbb{R}$ are given boundary functions.
\item $n_b$ is the number of given boundary functions.
\end{itemize}

Given a general PDE defined in Equation~\eqref{eq:PDE}, for all $k$, let $\left\{x_{b}^{(i)} ,g_{k}\big( x_{b}^{(i)}\big)\right\}_{i\in[1,N_{b}]}$ and $\left\{x_{d}^{(i)}\right\}_{i=1}^{N_{d}}$ be the sets of randomly selected boundary points and residual points for training, respectively. These points are usually drawn from an unknown distribution. Consider a computational model $\mathcal{Q}(\boldsymbol{\theta})$ that is parameterized by $\boldsymbol{\theta}\in \mathbb{R}^{m}$, typically represented by either a quantum or a classical neural network, where the quantum implementation is referred to as a quantum model. We define the hypothesis space:
$$
\mathcal{F}_\mathcal{Q}=\left\{f_{\mathcal{Q}(\boldsymbol{\theta})} :{\Omega}\cup \partial {\Omega}\rightarrow \mathbb{R}\right\}_{\boldsymbol{\theta} \in \mathbb{R}^{m}},
$$
where $f_{\mathcal{Q}(\boldsymbol{\theta})}$ denotes the function induced by the computational model $\mathcal{Q}(\boldsymbol{\theta})$, referred to as the hypothesis function that approximates the true solution $u$. The hypothesis function $f_{\mathcal{Q}(\boldsymbol{\theta})}$ represents the output of the computational model $\mathcal{Q}(\boldsymbol{\theta})$ as a function of the input domain $\Omega \cup \partial \Omega$ and the trainable parameters $\boldsymbol{\theta}$.

Then we define the loss function as
\begin{equation}
\label{eq:loss_function}
\mathcal{L}_{\text{total}}(\boldsymbol{\theta})=\sum_{k=1}^{n_d}\mathcal{L}_{k,b}(\boldsymbol{\theta})+\mathcal{L}_{d}(\boldsymbol{\theta}),
\end{equation}
where boundary loss values $\mathcal{L}_{k,b} (\boldsymbol{\theta} )$ and PDE residual loss value $\mathcal{L}_{d} (\boldsymbol{\theta})$ are defined as
\begin{align}
\label{eq:loss_details}
\mathcal{L}_{k,b} (\boldsymbol{\theta} )=\frac{W_{k}}{N_{b}}\sum _{i=1}^{N_{b}}\left| g_{k}( x_{b}^{(i)}) -f_{\mathcal{Q}(\boldsymbol{\theta})}\left( x_{b}^{(i)}\right)\right|^2 \notag \\
\mathcal{L}_{d} (\boldsymbol{\theta})=\frac{W_{d}}{N_{d}}\sum _{i=1}^{N_{d}}\left| D_{x}\left( f_{\mathcal{Q}(\boldsymbol{\theta})}\left( x_{d}^{(i)}\right) ;\lambda \right) -h\left( x_{d}^{(i)}\right)\right|^2,
\end{align}
with loss weight parameters $W_{k} ,W_{d} >0,\forall k\in[1,n_b]$. Specifically, the loss function for the HJB PDE in Equation  \ref{eq:HJB_PDE} is defined as:
\begin{equation}
\label{eq:HJB_loss}
\mathcal{L}_{\text{total}} (\boldsymbol{\theta})=\mathcal{L}_{1,b}(\boldsymbol{\theta})+\mathcal{L}_{2,b}(\boldsymbol{\theta})+\mathcal{L}_{d}(\boldsymbol{\theta}),
\end{equation}
For $W_d,W_1,W_2>0$, we have
\begin{equation}
\label{eq:HJB_loss_details}
\begin{aligned}
\mathcal{L}_{d}(\boldsymbol{\theta})&=\frac{W_{d}}{N_{d}}\sum_{i=1}^{N_{d}}
\Bigl(\partial_t v(t,x)\partial_x^2 v(t,x)+ \\
\partial_x v(t,x)&\partial_x^2 v(t,x)\,rx
-\frac{1}{2}\left(\frac{\mu-r}{\sigma}\right)^2\partial_x v(t,x)^2
\Bigr)^2,\\
\mathcal{L}_{1,b}(\boldsymbol{\theta})&=\frac{W_{1}}{N_{b}}\sum_{i=1}^{N_{b}}
\left(
f_{\mathcal Q}(T,x;\boldsymbol{\theta})-\frac{x^\gamma}{\gamma}
\right)^2,\\
\mathcal{L}_{2,b}(\boldsymbol{\theta})&=\frac{W_{2}}{N_{b}}\sum_{i=1}^{N_{b}}
\left(
f_{\mathcal Q}(t,1;\boldsymbol{\theta})-e^{-k(T-t)}\frac{1}{\gamma}
\right)^2,
\end{aligned}
\end{equation}
where $ k=-r\gamma+\frac{1}{2}\frac{\gamma}{\gamma-1}
\left(\frac{\mu-r}{\sigma}\right)^2$. With this in mind, we define QPINNs as follows:

\begin{definition}[Quantum Physics-Informed Neural Networks]
Given a PDE, we consider a quantum model $\mathcal{Q}(\boldsymbol{\theta})$ defined on the PDE domain, whose hypothesis function is denoted by $f_{\mathcal{Q}(\boldsymbol{\theta})}$, and define $\mathcal{F}_\mathcal{Q}$ as the corresponding hypothesis space. A Quantum Physics-Informed Neural Network (QPINN) seeks an approximation $f_{\mathcal{Q}(\boldsymbol{\theta}^*)}\in \mathcal{F}_\mathcal{Q}$ of the true solution $u$ by solving the optimization problem
\begin{equation}
\label{eq:PINN optimization}
\boldsymbol{\theta}^*=\arg\min_{\boldsymbol{\theta}\in\mathbb{R}^m}\ \mathcal{L}_{\text{total}}(\boldsymbol{\theta}),
\end{equation}
where $\mathcal{L}_{\text{total}}(\boldsymbol{\theta})$ is the loss function defined in Equations \ref{eq:loss_function} and \ref{eq:loss_details}. We refer to this framework as QPINNs.
\end{definition}
Replacing the quantum model with a classical one recovers the conventional PINN framework. In this study, we focus primarily on quantum models. We can define quantum models $\mathcal{Q}(\boldsymbol{\theta})$ on $n$ qubits as 
$$\mathcal{Q}(\boldsymbol{\theta}) = \bigl(U(\boldsymbol{x}, \boldsymbol{\theta}),O\bigr),$$
where the $2^n$-dimensional unitary $U$ is denoted as a parameterized quantum circuit (PQC), with the vector of trainable parameters $\boldsymbol{\theta} \in \mathbb{R}^{m}$, the classical data vector $\boldsymbol{x}=\left(x_1, \ldots, x_D\right) \in X \subset \mathbb{R}^D$, and $O$ is a Hermitian observable. We have the hypothesis function generated by $\mathcal{Q}(\boldsymbol{\theta})$ as
$$
 f_{\mathcal{Q}(\boldsymbol{\theta})}(\boldsymbol{x})=\langle 0| U(\boldsymbol{x}, \boldsymbol{\theta})^{\dagger} O U(\boldsymbol{x}, \boldsymbol{\theta})|0\rangle.
$$

\begin{figure}
    \centering
    \includegraphics[width=1\linewidth]{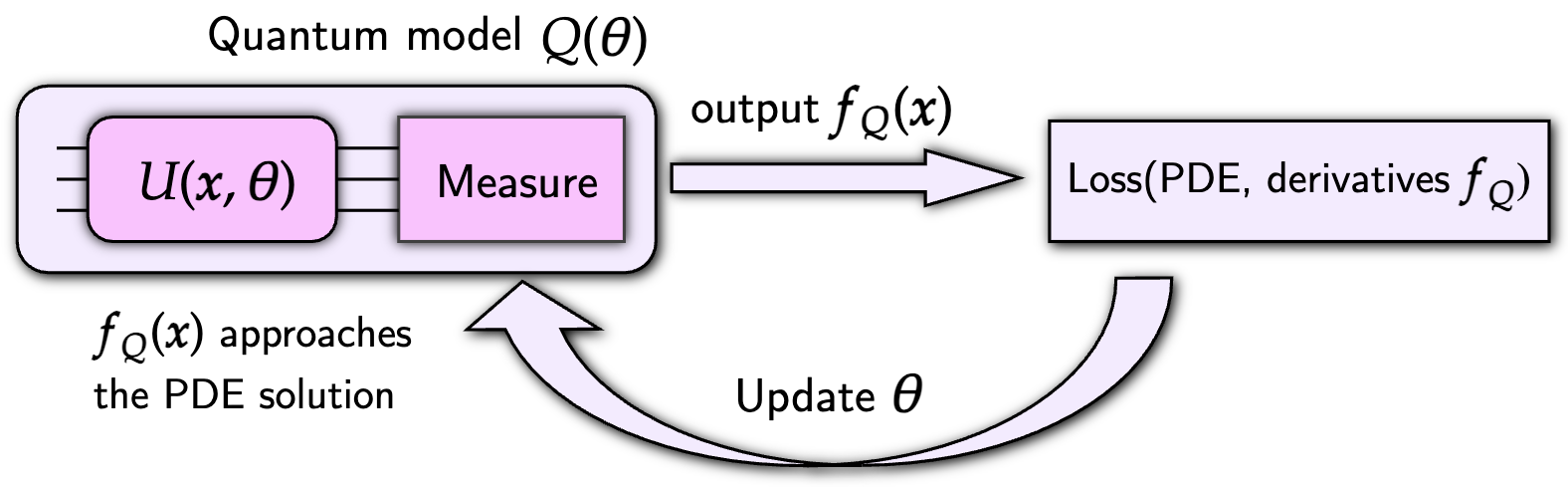}
    \caption{Workflow of Quantum Physics-Informed Neural Networks (QPINNs).}
    \label{fig:workflow_QPINNs}
\end{figure}
The workflow of QPINNs is shown in Figure \ref{fig:workflow_QPINNs}: we begin by choosing a quantum model $\mathcal{Q}(\boldsymbol{\theta})$ consisting of a PQC $U$ and an observable $O$, which together define the hypothesis function $f_{\mathcal{Q}(\boldsymbol{\theta})}$. Using the parameter shift rule \citep{Mitarai_2018}, we obtain estimated derivatives of $f_{\mathcal{Q}(\boldsymbol{\theta})}$. From these derivatives and guided by the PDE we want to solve, we evaluate a loss function defined in Equations~\ref{eq:loss_function} and \ref{eq:loss_details}, then apply gradient descent or other algorithms to update the circuit parameters $\boldsymbol{\theta}$. In principle, each update brings the model's output closer to the true solution of the PDE.

The Weierstrass approximation theorem guarantees that any smooth function can be approximated by a polynomial, and the approximation of multivariate functions is ensured by its multivariate extension~\citep{Nikolskii1991}. \citet{yu2024nonasymptotic} showed that with suitable choices of $\mathcal{Q}(\boldsymbol{\theta})$ and under ideal (noise-free and sufficiently large) conditions, any multivariate polynomial can be implemented. This implies that any continuous target function can be approximated to arbitrary accuracy by a quantum model employing polynomial hypothesis functions, or by a model whose hypothesis space contains arbitrary polynomials. However, implementing general polynomials requires exponentially deep circuits \citep{yu2024nonasymptotic}, and running such circuits on current noisy intermediate-scale quantum (NISQ) hardware is infeasible due to limited coherence times, gate fidelities, and qubit counts \citep{Preskill_2018}. Hence, we seek an efficient method to express suitable polynomials within a quantum model. In the following section, we present several classes of polynomials, including the tensor-decomposed polynomials, and describe how they can be implemented using quantum models.

\subsection{Polynomial variants}
\label{sec:polynomials summary}
We introduce two forms of polynomials: the general polynomial and its tensor-decomposed variant, highlighting how the latter enables efficient quantum implementation.

\begin{definition}[Polynomial]
\label{def:multivariate polynomial}
A real univariate polynomial with degree $L$, coefficient $c_{n}\in\mathbb{R}$, index $n\in[L]$ is defined as $p\left(x\right)=\sum_{n}c_{n}x^{n}$. A real multivariate polynomial $p(\boldsymbol{x})$ with $D$ variables and degree at most $L$ in each variable is defined as
\begin{equation}
\label{eq:polynomial}
p(\boldsymbol{x}):=\sum_{\boldsymbol{n} } c_{\boldsymbol{n}} \boldsymbol{x}^{\boldsymbol{n}},
\end{equation}
where $\boldsymbol{x}=\left(x_1, \ldots, x_D\right) \in \mathbb{R}^D, \boldsymbol{n}=\left(n_1, \ldots, n_D\right) \in [L]^D, c_{\boldsymbol{n}} \in \mathbb{R}$ and $\boldsymbol{x}^{\boldsymbol{n}}=x_1^{n_1} x_2^{n_2} \cdots x_D^{n_D}$.
\end{definition}
We define polynomials with degree at most $L$ in each variable as $\|\boldsymbol{n}\|_\infty=\max _i n_i \leq L$, polynomials with total degree at most $L$ as $\|\boldsymbol{n}\|_1=\sum_i n_i \leq L$, and $[N]:=\{0,\cdots,N\}$ for $N\in \mathbb{Z^+}$. 
In addition to this ordinary form, we can apply the \textit{tensor rank decomposition} \citep{kolda2009tensor} to the polynomial coefficients such that
\begin{equation}
\label{eq:tensor_general_c}
c_{\boldsymbol{n}} = c_{n_1,\dots,n_D}
= \sum_{r=1}^R
\lambda_r c_{r,n_1} c_{r,n_2} \cdots c_{r,n_D},
\end{equation}
where $c_{r,n_j}\in\mathbb{R}$, $n_j\in[L]$, $\forall r\in[R], j\in[D]$, and the tensor rank $R\in[1,(L+1)^D]$. By combining Equations \ref{eq:polynomial} and \ref{eq:tensor_general_c}, we obtain the following definition:

\begin{definition}[Tensor-Decomposed Polynomial]
\label{def:TD poly}
Let $p(\boldsymbol{x})$ be a real multivariate polynomial with $D$ variables and degree at most $L$ in each variable, it can also be written as
\begin{equation}
\label{eq:tensor_poly}
p(\boldsymbol{x})=\sum_{r=1}^R \lambda_r \prod_{j=1}^D p_{r,j}(x_j),
\end{equation}
where each $p_{r, j}\left(x_j\right)=\sum_{n_j}c_{r,n_j}x_j^{n_j}$ is a univariate polynomial, the tensor rank is $R\in[1,(L+1)^D]$, $\lambda_r\in\mathbb{R}$, $c_{r,n_j}\in\mathbb{R}$, and $n_j\in[L],\forall r\in[R],j\in[D]$. We call such $p(\boldsymbol{x})$ a tensor-decomposed polynomial (TD polynomial).
\end{definition}
See details in Appendix \ref{sec: TD_polynomials}. \LW{In structured settings, a PDE solution may be well-approximated by a hypothesis of relatively low rank $R$}. One may fix $R$ in advance and then train to obtain a solution within an acceptable error, which corresponds to the notion of canonical polyadic (CP) decomposition.

\section{Main results}
Having defined the polynomial and its tensor-decomposed form, we discuss the quantum circuit implementation of them. We first present a summary table outlining our main theoretical results, and then provide detailed derivations and explanations in the subsequent sections.

\subsection{Polynomial resource summary}
\label{sec: summary}

\setlength{\intextsep}{10pt} 
\begin{table*}[ht]
    \centering
    {\huge  
    \renewcommand{\arraystretch}{1.8}
    \begin{adjustbox}{max width=\textwidth} 
        \begin{tabular}{cllll} \specialrule{2.0pt}{0pt}{0pt}
      Polynomial form &Width& Depth&Number of parameters  &Reference\\ \specialrule{1.6pt}{0pt}{0pt}
 \makecell{Univariate polynomial \\ with $ |n| \le L $}& $3$& $O(L)$& $2L+1$&Our Proposition \ref{prop:univariat_poly}\\\hline
 \makecell{Multivariate polynomial \\ with $\|\boldsymbol{n}\|_1 \le L$}&$O(D+\log L+L \log D)$&  $O\left(L^2 D^L(\log L+L \log D)\right)$& $O\left(L D^L(L+D)\right)$&\citet[Theorem 1]{yu2024nonasymptotic}\\\hline
 \makecell{Multivariate polynomial \\ with $\|\boldsymbol{n}\|_{\infty} \le L$}& $O(D \log L)$& $O\left(D^2 L^{D+1} \log L\right)$& $O(D L^{D+1})$&Our Theorem \ref{theorem:multivariate_uniform}\\\hline
 \makecell{TD polynomial \\ with $R\in[1,(L+1)^D]$}& $2D+\lceil\log R\rceil+1$&  $O(RDL \log R)$&  $O\left(RDL\right)$&Our Theorem \ref{theorem:tensor_poly}\\\hline
        \end{tabular}
    \end{adjustbox}
    }
    \caption{Summary of the quantum model resources required for implementing polynomials presented in Definitions \ref{def:multivariate polynomial} and  \ref{def:TD poly}.}
    \label{table:comparison}
\end{table*}

Table \ref{table:comparison} summarizes both known complexity bounds \citep{yu2024nonasymptotic} and our own results (marked as propositions/theorems in later sections). Specifically, “Width” refers to the number of qubits required (i.e., the quantum circuit register size); “Depth” denotes the circuit depth under the native gate set, which includes single-qubit gates and CNOT gates; “Number of parameters” indicates the total number of trainable circuit parameters. The quantum model output is always real and its absolute value is bounded by $1$, since it is defined through the expectation values of Hermitian observables whose eigenvalues lie in $[-1,1]$. A scaling factor is typically applied to ensure accurate approximation of target functions. See the references in the table for details. 

It should be noted that the required resource complexity for different quantum computer platforms could be different. For some types of quantum computer, such as those based on the neutral atom platform, multi‐controlled rotation gates can be implemented directly and relatively efficiently \citep{PhysRevA.103.062607}. In this case, they are considered as “native gates” for neutral atom quantum computers. In contrast, for platforms such as superconducting quantum computers, multi‐controlled rotation gates must be decomposed into a large number of simpler "native gates" such as CNOT and single-qubit gates. It is important to note that the cost of implementing the same quantum gate can vary significantly across different quantum computing platforms, therefore we analyze the implementation cost complexity of quantum models from the perspective of different native gate sets. In this work, we consider cases where multi-controlled gates are treated as native gates only when explicitly stated. Otherwise, the native gate set consists of single-qubit gates and CNOT gates (e.g., Table \ref{table:comparison}).

The quantum model resources needed depend on different polynomial forms according to their algebraic structure (ordinary or tensor-decomposed form) and their norm constraints ($\|\boldsymbol{n}\|_1 \leq L$ or $\|\boldsymbol{n}\|_\infty \leq L$). The choice of the $\|\boldsymbol{n}\|_1$ norm constraint leads to combinatorial growth $O(D^L)$, while the $\|\boldsymbol{n}\|_\infty$ norm constraint produces $O(L^D)$ terms, changing the complexity in both quantum circuit depth and number of parameters. Moreover, for polynomials that can exploit tensor rank decomposition to factorize multivariate coefficients into a sum of products of univariate coefficients in Equation~\eqref{eq:tensor_general_c}, we can significantly optimize the complexity with respect to $D$ compared to \citet[Theorem 1]{yu2024nonasymptotic}, as both the circuit depth and the number of parameters are reduced from exponential to polynomial complexity when the tensor rank is not large. In the following sections \ref{sec:poly} and \ref{sec: TD poly}, we present the detailed descriptions of the results summarized in Table \ref{table:comparison}.

\subsection{Polynomial implementation}
\label{sec:poly}
In this subsection, we present how to implement polynomials within quantum models. 
\subsubsection{Univariate case}
We first introduce the quantum model for implementing real univariate polynomials, based on QSP lemmas from \citet{Gily_n_2019}. Given an integer $L$ and parameter $\boldsymbol{\theta} \in \mathbb{R}^{L+1}$, with input $x\in[-1,1]$, define the parameterized unitary
\begin{equation}
    \label{eq:u_theta_main}
    U_{\boldsymbol{\theta}}(x):=R_z\left(\theta_L\right) \prod_{j=0}^{L-1} R_x(-2 \arccos (x)) R_z\left(\theta_j\right),
\end{equation}
that is shown in Figure \ref{fig:single_qubit_QSP_main}. Here, $R_z(\theta)$ and $R_x(\theta)$ denote single-qubit rotation gates about the Pauli-Z and Pauli-X axes, respectively, each parameterized by the rotation angle $\theta$.
\begin{figure}[H]
    \centering
    \includegraphics[width=1.0\linewidth]{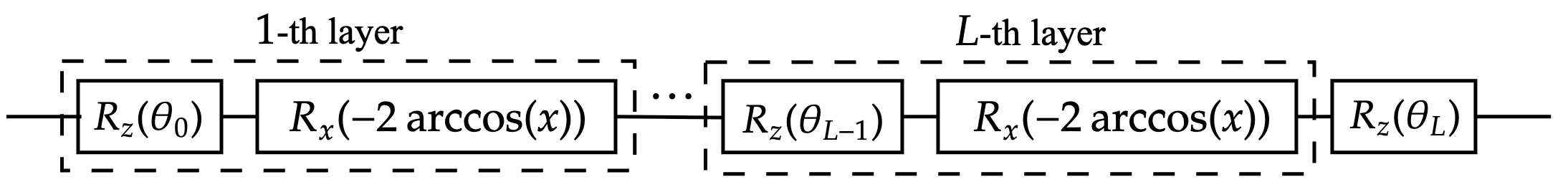}
    \caption{Circuit of $U_{\boldsymbol{\theta}}(x)$. Since the circuit has an alternating structure, we consider it to comprise $L$ layers, each consisting of $R_x(-2 \arccos (x))$ and $R_z\left(\theta_j\right)$, followed by an additional $R_z\left(\theta_L\right)$.}
    \label{fig:single_qubit_QSP_main}
\end{figure}

We then propose the quantum resource bound analysis for the implementation of arbitrary real univariate polynomial under two different quantum native gate sets:
\begin{proposition}[Quantum circuit for univariate polynomial]
\label{prop:univariat_poly}
For any real polynomial $p(x) \in \mathbb{R}[x]$ that satisfies $\operatorname{deg}(p(x)) \leq L$ and $\forall x \in[-1,1],|p(x)| \leq \frac{1}{2}$, there exists a quantum model $\mathcal{Q}$ that consists of a PQC $W_p(\boldsymbol{x})$ and an observable $Z^{(0)}$ such that 
$$
f_{\mathcal{Q}}(x):=\langle 0| W_p^{\dagger}(x) Z^{(0)} W_p(\boldsymbol{x})|0\rangle=p(x),
$$
where $Z^{(0)}$ is the Pauli $Z$ observable on the first qubit. The width of the PQC is at most $3$, the number of parameters is at most $2L+1$. The PQC $W_p(\boldsymbol{x})$ can be expressed as at most $4L$ double‐controlled rotation gates and $4$ Hadamard gates, with depth $4L+2$. Alternatively, it can be expressed using $36L$ single‐qubit gates and $32L$ CNOT gates, with depth $60L-5$.
\end{proposition}

The circuit of $W_p(\boldsymbol{x})$ is shown in Figure \ref{fig:single_qubit_poly_main}.
\begin{figure}[H]
    \centering
    \includegraphics[width=0.6\linewidth]{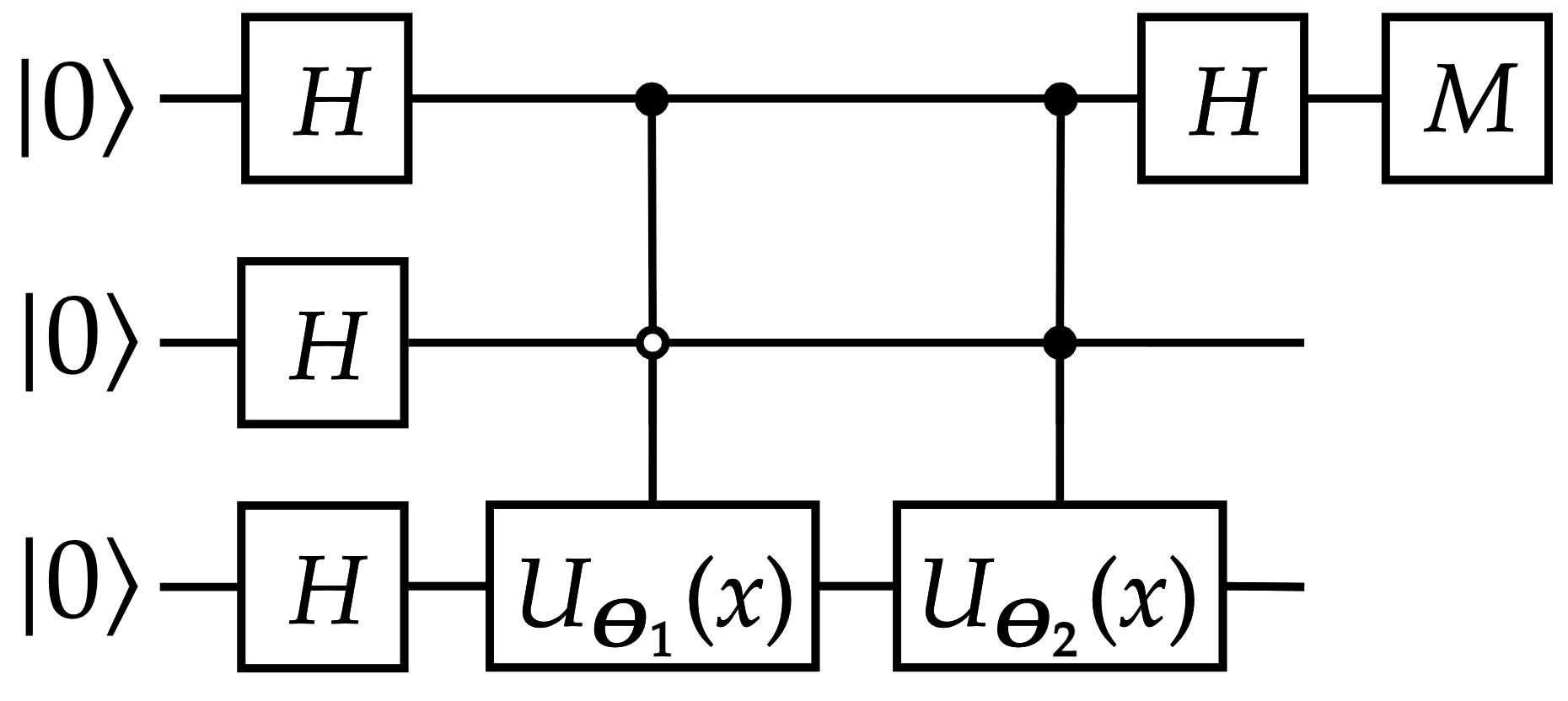}
    \caption{The quantum model of Proposition \ref{prop:univariat_poly} can be considered as the Hadamard test to estimate $p(x)=\langle+|^{\otimes2} (|0\rangle\langle0|\otimes U_{\boldsymbol{\theta_{1}}}(x)+|1\rangle\langle1|\otimes U_{\boldsymbol{\theta_{2}}}(x))|+\rangle^{\otimes2}$ where  $U_{\boldsymbol{\theta_{1}}}(x),U_{\boldsymbol{\theta_{2}}}(x)$ are defined in Equation~\eqref{eq:u_theta_main} and shown in Figure \ref{fig:single_qubit_QSP_main}. $H$ and $M$ denote the Hadamard gate and measurement operation, respectively; this notation applies throughout this work. The hollow/open circle denotes a negative control, meaning that the controlled unitary is applied when the control qubit is $|0\rangle$. By contrast, the filled/solid black circle denotes the usual control on  $|1\rangle$, as in a standard controlled gate.}
    \label{fig:single_qubit_poly_main}
\end{figure}

The proof is provided in Appendix \ref{sec:univariat_poly}. As discussed in Section \ref{sec: summary}, the required quantum resource complexity can vary across different quantum computing platforms. The resource analysis with double-controlled rotation gates as native gates corresponds to neutral-atom quantum platforms, whereas the one considering only single-qubit and CNOT gates as native gates corresponds to superconducting platforms.

\subsubsection{Multivariate case}
We consider the implementation of multivariate polynomials. Given an integer $L$ and parameter $\boldsymbol{\theta} \in \mathbb{R}^{L+1}$, with input $x\in[-1,1]$, define the parameterized unitary
\begin{equation}
\begin{aligned}
\label{eq:monomial_main}
&U^{p}(\boldsymbol{x})=\\
&\bigotimes_{i=1}^D \left(R_z\left(\theta^{i}_L\right) \prod_{j=0}^{L-1} R_x(-2 \arccos (x)) R_z\left(\theta^{i}_j\right)\right),
\end{aligned}
\end{equation}

that is shown in Figure \ref{fig:U^p_main}.
\begin{figure}[H]
    \centering
    \includegraphics[width=1.0\linewidth]{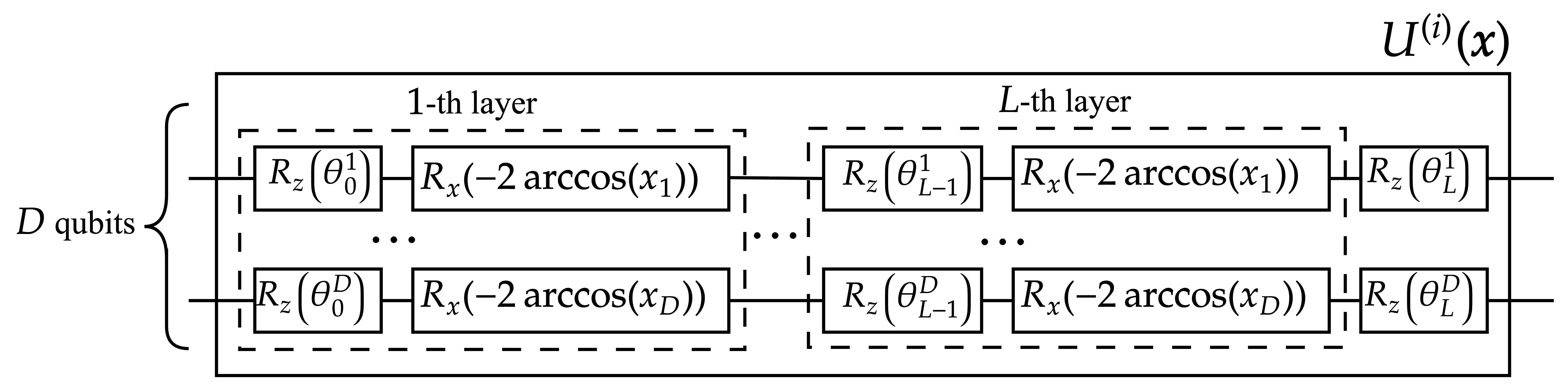}
    \caption{The circuit $U^p(\boldsymbol{x})$ retains exactly the same structure as $U_{\boldsymbol{\theta}}(x)$ in Figure \ref{fig:single_qubit_QSP_main}, except that it acts on multiple qubits via tensor products.}
    \label{fig:U^p_main}
\end{figure}
We can then propose the quantum resource complexity analysis for implementing any real multivariate polynomial:

\begin{theorem}[Quantum circuit for multivariate polynomial]
\label{theorem:multivariate_uniform}
Let  $p(\boldsymbol{x})=\sum_{\boldsymbol{n} } c_{\boldsymbol{n}} \boldsymbol{x}^{\boldsymbol{n}}$ be any real multivariate polynomial with $D$ variables and degree at most $L$ in each variable such that $\boldsymbol{x} \in[-1,1]^D$, $c_{\boldsymbol{n}}\in \mathbb R$, $\boldsymbol{n} \in[L]^D$. Then there exists a quantum model $\mathcal{Q}$ that consists of a PQC $W_p(\boldsymbol{x})$ and an observable $Z^{(0)}$ with the scaling factor $\Lambda$ such that
$$
f_{\mathcal{Q}}(\boldsymbol{x}):=\langle 0| W_p^{\dagger}(\boldsymbol{x}) Z^{(0)} W_p(\boldsymbol{x})|0\rangle=p(\boldsymbol{x})/\Lambda,
$$
where $\Lambda=|c_{\boldsymbol{n}}|_\infty  (L+1)^D$ and $Z^{(0)}$ is the Pauli $Z$ observable on the first qubit. The width of the PQC is $O(D \log L)$, the depth is $O\left(D^2 L^{D+1} \log L\right)$, and the number of parameters is $O\left(D L^{D+1}\right)$.
\end{theorem}

The circuit diagram of $W_p(\boldsymbol{x})$ is shown in Figure \ref{fig:Circuit_diagram_withoutT_main}.

\begin{figure}[H]
    \centering
    \includegraphics[width=1\linewidth]{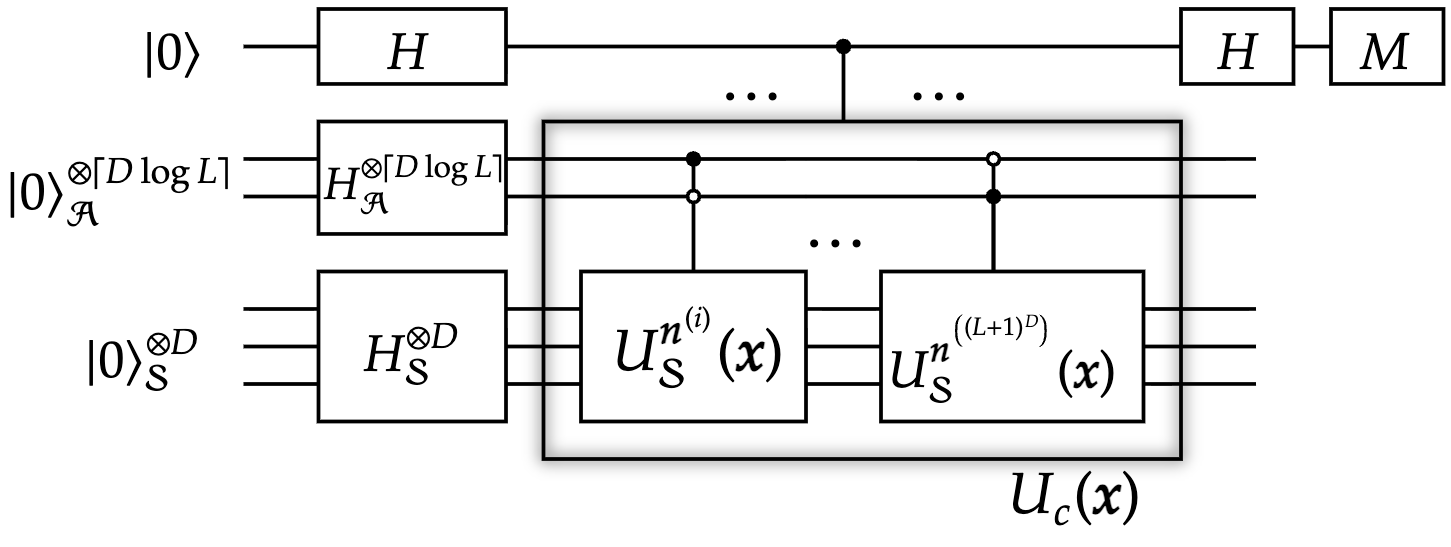}
    \caption{The circuit of $ W_p(\boldsymbol{x})$ consists of a single-qubit system, a $\lceil D\log L\rceil$-qubit system $\mathcal{A}$, and a $D$-qubit system $\mathcal{S}$. This quantum model can be viewed as a Hadamard test to estimate $p(\boldsymbol x)=\langle+|^{\otimes \lceil D\log L\rceil}_{\mathcal{A}}
 \langle+|^{\otimes D}_{\mathcal{S}}
 U_c(\boldsymbol{x})
 |+\rangle^{\otimes \lceil D\log L\rceil}_{\mathcal{A}}
 |+\rangle^{\otimes D}_{\mathcal{S}}$ such that $U_c(\boldsymbol{x})=\sum_{i=1}^{(L+1)^D}\ket{i}\bra{i}_\mathcal{A}\otimes U^{\boldsymbol{n}^{(i)}}_\mathcal{S}(\boldsymbol{x})$ where $U^{\boldsymbol{n}^{(i)}}_\mathcal{S}(\boldsymbol{x})$ is defined in Equation~\eqref{eq:monomial_main} and shown in Figure \ref{fig:U^p_main}. }
    \label{fig:Circuit_diagram_withoutT_main}
\end{figure}
The proof is provided in Appendix \ref{sec:multivariate_uniform}. Theorem \ref{theorem:multivariate_uniform} is closely related to \citet[Theorem 1]{yu2024nonasymptotic}, as both establish circuit complexity bounds for real multivariate polynomials. However, the exact complexities differ, and we additionally provide an explicit bound on the scaling factor for general real multivariate polynomials. Despite this, both their circuit depth and number of parameters grow exponentially, making them prohibitively expensive in practice. In the next part, we propose a new theorem within the tensor decomposition framework, showing that under suitable constraints, some circuit resources (width, depth, number of parameters) exhibit only polynomial complexity, rendering the approach far more practical for real-world problems.

\subsection{Tensor-decomposed polynomial implementation}
\label{sec: TD poly}

In this subsection, we present the quantum models designed to implement tensor-decomposed polynomials and analyze their corresponding quantum resource complexity.
\begin{theorem}[Quantum circuit for tensor-decomposed polynomial]
\label{theorem:tensor_poly}
Let $p(\boldsymbol{x})$ be any real multivariate polynomial with $D$ variables and degree at most $L$ in each variable such that 
$$
p(\boldsymbol{x})=\sum_{i=1}^R \lambda_i \prod_{j=1}^D p_{i, j}\left(x_j\right),
$$
where $R\in[1,(L+1)^D]$, $\sum_{i=1}^R |\lambda_i|=\Lambda\in\mathbb{R}$, and each $p_{i,j}(x_j)$ is a univariate polynomial of degree $L$ satisfying $|p_{i,j}(x_j)| \leq 1/2$  for $x_j \in[-1,1],\forall i\in[R],j\in[D]$. Then, there exists a quantum model $\mathcal{Q}$ that consists of a PQC $W_p(\boldsymbol{x})$ and an observable $Z^{(0)}$ such that
$$
f_{\mathcal{Q}}(\boldsymbol{x}):=\langle 0| W_p^{\dagger}(\boldsymbol{x}) Z^{(0)} W_p(\boldsymbol{x})|0\rangle=p(\boldsymbol{x})/\Lambda,
$$
where $Z^{(0)}$ is the Pauli $Z$ observable on the first qubit. The width of the PQC is at most $2D+\lceil\log R\rceil+1$, the depth is $O(RDL \log R)$, and the number of parameters is $O\left(RDL\right)$.
\end{theorem}
The circuit diagram of $W_p(\boldsymbol{x})$ is shown in Figure \ref{fig:theorem2_circuit_main} and  Figure \ref{fig:V_circuit_main}. 

\begin{figure}[H]
    \centering
    \includegraphics[width=1\linewidth]{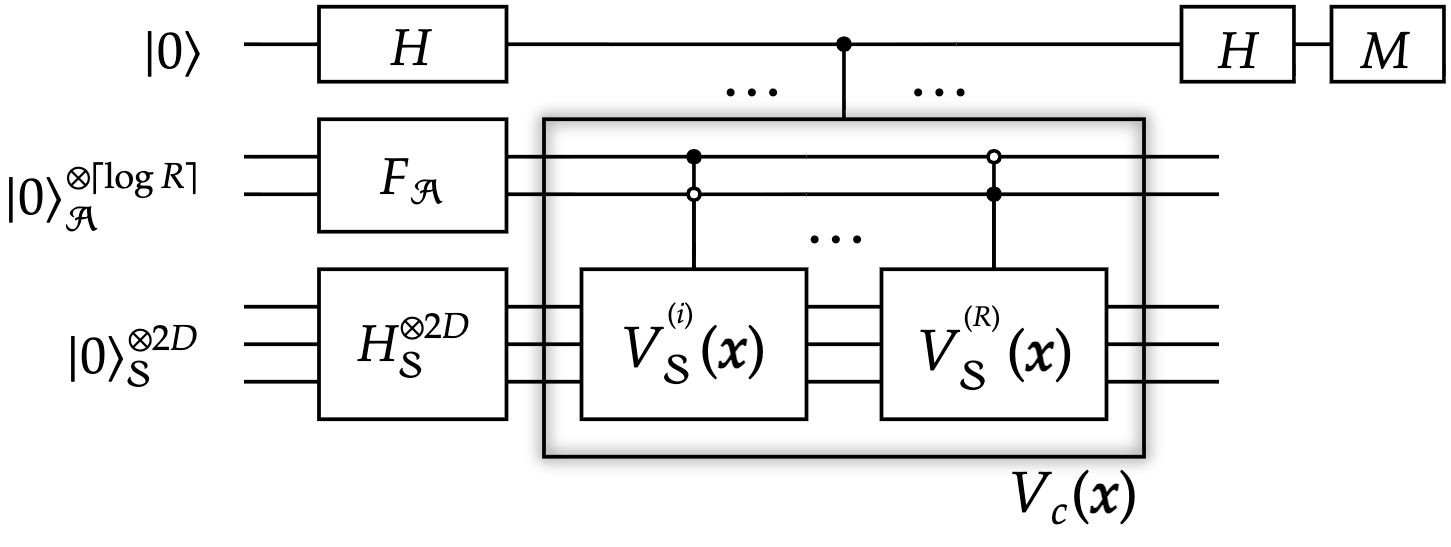}
    \caption{The circuit of $ W_p(\boldsymbol{x})$ consists of a single-qubit system, a $\lceil \log R\rceil$-qubits system $\mathcal{A}$, and a $2D$-qubits system $\mathcal{S}$. We have $F_\mathcal{A}\,\ket{0}^{^{\otimes\lceil\log R\rceil}}_\mathcal{A}
=\sum_{i=1}^R\sqrt{\frac{\lvert \lambda_i\rvert}{\Lambda}}|i\rangle_\mathcal{A}$, and this quantum model can be considered as the Hadamard test circuit to estimate $p(\boldsymbol x)/\Lambda=\left\langle v_p\right| V_c(\boldsymbol{x})\left|v_p\right\rangle$ such that $\left|v_p\right\rangle=\sum_{i=1}^R \sqrt{\frac{\lvert \lambda_i\rvert}{\Lambda}}|i\rangle_{\mathcal{A}} \otimes|+\rangle_{\mathcal{S}}^{\otimes 2 D}$, and $V_c(\boldsymbol{x})
=\sum_{i=1}^R \ket{i}\bra{i}_\mathcal{A}\;\otimes\;V^{(i)}_\mathcal{S}(\boldsymbol{x})$ where $V^{(i)}_\mathcal{S}(\boldsymbol{x})$ is shown in Figure \ref{fig:V_circuit_main}.}
    \label{fig:theorem2_circuit_main}
\end{figure}

\begin{figure}[H]
    \centering
    \includegraphics[width=0.85\linewidth]{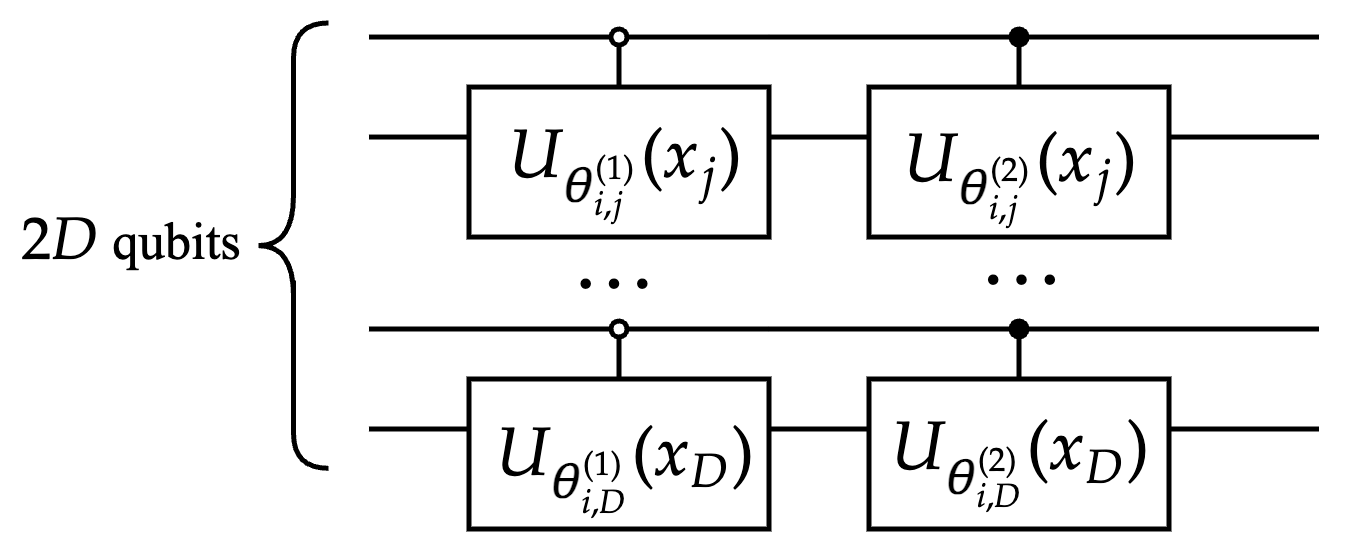}
    \caption{Circuit of $V^{(i)}_\mathcal{S}(\boldsymbol{x})
=\bigotimes_{j=1}^D
\Bigl(\ket{0}\!\bra{0}\otimes U_{\boldsymbol{\theta}_{i,j}^{(1)}}(x_j)
+\ket{1}\!\bra{1}\otimes U_{\boldsymbol{\theta}_{i,j}^{(2)}}(x_j)\Bigr)$ where $U_{\boldsymbol{\theta}_{i,j}^{(1)}}(x_j),U_{\boldsymbol{\theta}_{i,j}^{(2)}}(x_j)$ are defined in Equation~\eqref{eq:u_theta_main} and shown in Figure \ref{fig:single_qubit_QSP_main}. }
    \label{fig:V_circuit_main}
\end{figure}

The proof is provided in Appendix \ref{sec_app: TD poly}, and we refer to this quantum model as the \textit{tensor-decomposed model}. We can observe that the circuit depth and the number of parameters grow polynomially rather than exponentially when the tensor rank $R$ is not very large.

\LW{To validate the constructive implementation in Theorem \ref{theorem:tensor_poly}, we consider the following explicit two-dimensional tensor-decomposed polynomial:
\begin{equation}
\label{eq:theorem3_td_qnn_target}
\begin{aligned}
p(x_1,x_2)=&
0.35\,(x_1^2-\tfrac12)(2x_2^3-\tfrac32x_2)\\
&+0.25\,(2x_1^3-\tfrac32x_1)(\tfrac12x_2)\\
&+0.25\,(\tfrac12)(x_2^2-\tfrac12)\\
&+0.15\,(\tfrac12x_1)(2x_2^3-\tfrac32x_2).
\end{aligned}
\end{equation}
This polynomial has $D=2$, $R=4$, univariate degree bound $L=3$, and $\Lambda=1$. Its univariate factors satisfy the condition $|p_{i,j}(x_j)|\leq 1/2$ on $[-1,1]$. We trained the model from Theorem \ref{theorem:tensor_poly} on 2000 uniformly sampled points in $[-1,1]^2$ using the standard mean squared error (MSE) loss and Adam with learning rate $10^{-1}$ for 2500 epochs. On a $10^3\times10^3$ evaluation grid, the relative $L^2$ error is $1.06\times10^{-5}$ and the maximum absolute error is $9.01\times10^{-6}$. Figure \ref{fig:theorem3_td_qnn_validation} shows that the trained quantum model reproduces the target surface, which provides a direct validation of the constructive statement in Theorem \ref{theorem:tensor_poly}.}

\begin{figure}[th]
    \centering
    \includegraphics[width=1\linewidth]{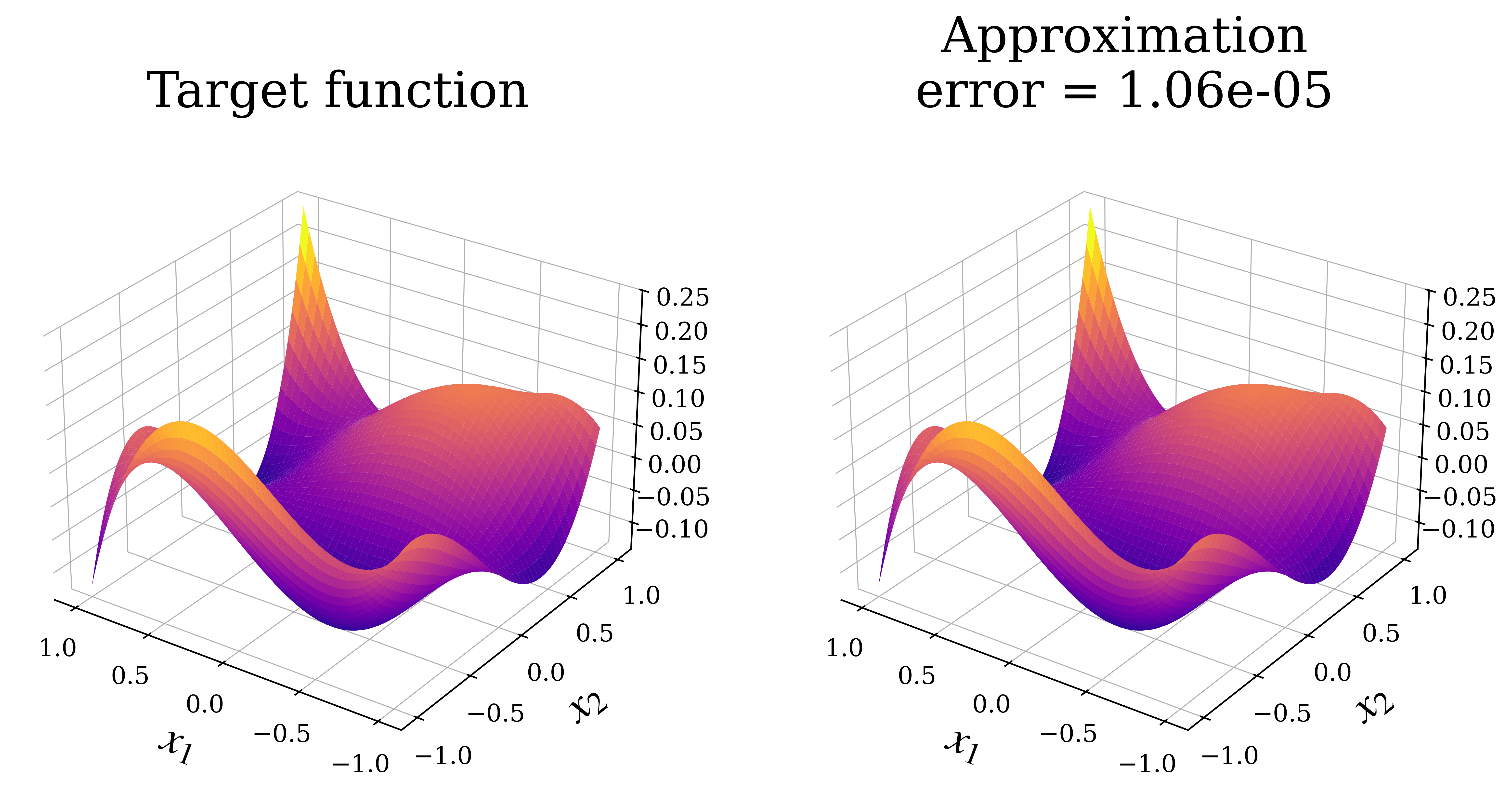}
    \caption{Numerical validation of Theorem \ref{theorem:tensor_poly} on an explicit TD polynomial with $D=2$, $R=4$, $L=3$, and $\Lambda=1$. The left panel shows the target function which is defined in Equation~\eqref{eq:theorem3_td_qnn_target}, and the right panel shows the output of the trained model from Theorem \ref{theorem:tensor_poly}. The relative $L^2$ error on the evaluation grid is $1.06\times10^{-5}$.}
    \label{fig:theorem3_td_qnn_validation}
\end{figure}

\LW{Here $L=3$ bounds the degree of each univariate factor in Theorem \ref{theorem:tensor_poly}, whereas the degree parameter in \citet[Theorem 1]{yu2024nonasymptotic} bounds the total degree through $\|\boldsymbol{n}\|_1$. Using the resource bounds in Table \ref{table:comparison}, Theorem \ref{theorem:tensor_poly} gives a tensor-decomposed model with a qubit bound of $7$, a leading depth quantity of $48$, and a leading parameter quantity of $24$. After expansion, the polynomial in Equation~\eqref{eq:theorem3_td_qnn_target} has maximum total degree $L_{\mathrm{tot}}=\|\boldsymbol{n}\|_1=5$. Applying \citet[Theorem 1]{yu2024nonasymptotic} to the same polynomial gives a qubit bound of $10$, a leading depth quantity of $6400$, and a leading parameter quantity of $1120$. This comparison shows that our tensor-decomposed model leads to a significant reduction in the required quantum resources, especially in circuit depth and trainable parameters.}

\LW{Although Theorem \ref{theorem:tensor_poly} is stated with a single degree parameter $L$, this parameter should be read as a common upper bound on the degrees assigned to individual variables. Since each univariate factor $p_{i,j}(x_j)$ is implemented by its own block in Figure~\ref{fig:V_circuit_main}, the construction can use different degrees $L_j$ for different variables, such as degree $2$ for one variable and degree $3$ for another. This degree choice concerns the univariate factors, while the tensor rank $R$ controls the number of product terms.} We now state the rank-one case of Theorem \ref{theorem:tensor_poly}.
\begin{corollary}[Quantum circuit for rank-1 tensor-decomposed polynomial]
\label{coro:rank1_poly}
Let  $p(\boldsymbol{x})=\prod_{j=1}^D p_{ j}\left(x_j\right)$ be any real multivariate polynomial such that $\forall j\in[D],x_j\in[-1,1]$, each $p_{i,j}(x_j)$ is a univariate polynomial of degree $L$ satisfying $|p_{i,j}(x_j)| \leq 1/2$. Then there exists a quantum model $\mathcal{Q}$ that consists of a PQC $W_p(\boldsymbol{x})$ and an observable $Z^{(0)}$ such that
$$
f_{\mathcal{Q}}(\boldsymbol{x}):=\langle 0| W_{p}^{\dagger}(\boldsymbol{x}) Z^{(0)}W_{p}(\boldsymbol{x})|0\rangle=p(\boldsymbol{x}),
$$
where $Z^{(0)}$ is the Pauli $Z$ observable on the first qubit. The width of the PQC is at most $2D+1$, the depth is at most $4LD+2$ by allowing double-controlled rotation gates to be implemented natively, and the number of parameters is at most $(2L+1)D$.
\end{corollary}

The proof is provided in Appendix \ref{section:rank1_poly}. Actually, the scaling bound for the univariate polynomial $|p_{i,j}(x_j)| \leq 1/2$ in Theorem \ref{theorem:tensor_poly} and Corollary \ref{coro:rank1_poly} is a sufficient but not a necessary condition for the existence of the quantum model $\mathcal{Q}$. The scaling bound for both results is derived from Corollary \ref{coro:univariate_poly_coro} in Appendix, where this condition was originally introduced. In fact, the necessary condition is given in Equation  \ref{eq:true_bound}. 

We apply the model based on Corollary \ref{coro:rank1_poly} to solve the Merton portfolio optimization problem, see Section \ref{sec:experiment} for details.

\subsection{Our QPINNs}
\label{sec:our_QPINN}
\LW{The tensor-decomposed model in Theorem~\ref{theorem:tensor_poly} explicitly implements tensor-decomposed polynomials, which are useful for PDE settings whose solutions admit exact or low-rank decomposable representations.} In this subsection, we first introduce a QPINN based on this tensor-decomposed model with an additional entangling layer. It is also useful to consider the QPINN without an entangling layer, which we refer to as the quantum-inspired PINN, meaning that it can be efficiently simulated classically. The general QPINN framework is defined in Section \ref{sec:QPINN}.


\subsubsection{QPINNs with entanglement}
\label{sec:our_realQPINN}
If we directly adapt the tensor-decomposed model within the QPINN framework, the corresponding hypothesis space contains tensor-decomposed polynomials with tensor rank $R$. However, since the quantum resource cost in Theorem \ref{theorem:tensor_poly} grows linearly with $R$, only relatively small values of $R$ are feasible on current quantum hardware. Consequently, restricting the model to a fixed tensor rank $R$ limits its expressivity and may lead to unacceptable approximation errors in complex cases.

To obtain a QPINN with better expressivity, we expand each block $V_c(\boldsymbol{x})$ defined in the circuit diagram of Theorem \ref{theorem:tensor_poly} by attaching an entangling unitary
\begin{equation}
\label{eq:add_layer}
V_c(\boldsymbol{x})
\;\longrightarrow\;
V_c(\boldsymbol{x})\, e^{-iH(\boldsymbol{\lambda})},
\end{equation}
where \(H(\boldsymbol{\lambda})\) is the \textit{entangling Hamiltonian}, parameterized by $\boldsymbol{\lambda}$ and acting on the \(2D\) variable qubits, satisfying
\begin{equation}
\label{eq:add_hamiltonian}
e^{-iH(\boldsymbol{\lambda}_0)}=I
\end{equation}
for a particular parameter choice $\boldsymbol{\lambda}_0$, where $D$ is the number of variables and $I$ is the identity matrix. This guarantees that the original tensor-decomposed hypothesis space is entirely contained within the enlarged hypothesis space with $e^{-iH(\boldsymbol{\lambda})}$, so the enlarged hypothesis space contains the original one and can be larger. Because of the circuit architecture shown in Figure \ref{fig:theorem2_circuit_main}, the entangling unitary $e^{-i H(\boldsymbol{\lambda})}$ is implemented in the overall circuit as a controlled $e^{-i H(\boldsymbol{\lambda})}$ operation with a single control qubit (typically the first qubit), in the same sense as the control structure used in CNOT gates. We refer to this controlled $e^{-iH(\boldsymbol{\lambda})}$ as the \emph{added entangling layer} in the QPINN. Figure~\ref{fig:circuit_added_entangle} shows its general implementation.

\begin{figure}[h]
    \centering
    \includegraphics[width=1\linewidth]{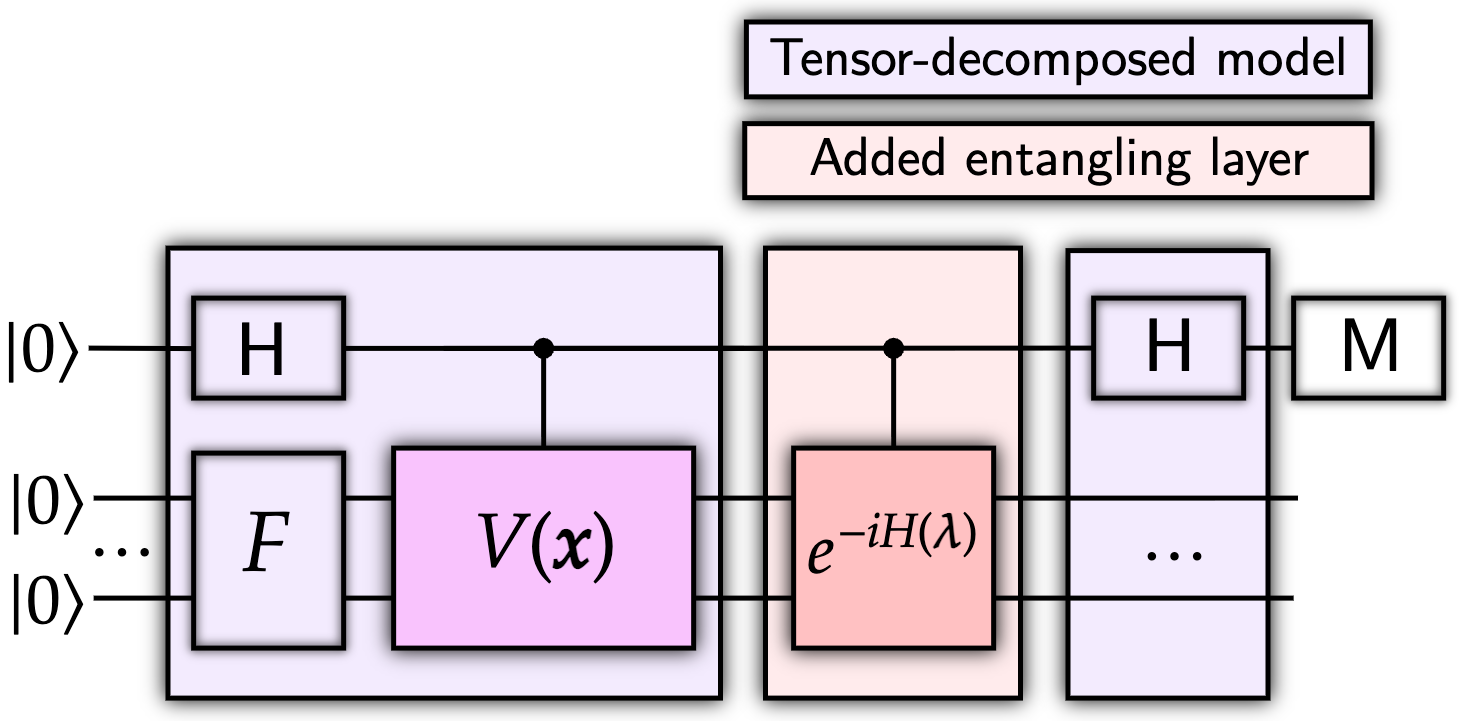}
    \caption{When the added entangling layer is the identity, this circuit is equivalent to the Hadamard test and recovers the tensor-decomposed model in Theorem \ref{theorem:tensor_poly}. $\langle v|\,V(\boldsymbol{x})\,|v\rangle=p(\boldsymbol{x})/\Lambda$ where $F|0\rangle=|v\rangle$.}
    \label{fig:circuit_added_entangle}
\end{figure}

The added entangling layer can introduce non-separable multi-qubit correlations, expanding the hypothesis space beyond tensor-decomposed structures and \LW{making the model no longer efficiently simulable by the factorized classical procedure used for the Quantum-inspired PINN.} This enlarged hypothesis space can provide additional expressivity and may help the QPINN approximate more complex PDE solutions. \LW{At the same time, an unsuitable added entangling layer may weaken the useful inductive bias inherited from the tensor-decomposed structure and lead to worse empirical performance.} However, increased expressivity may come at the cost of reduced trainability, particularly through the barren plateau (BP) phenomenon, which refers to the exponential vanishing of gradient variance during training \citep{McClean_2018}. Appropriate designs of the entangling Hamiltonian $H$ may help balance this trade-off and avoid the BP phenomenon. Further exploration of such designs is an interesting direction for future work.

\subsubsection{Quantum-inspired PINN}
\label{sec:QIPINN}
Although quantum models are usually intended to run on quantum computers, \citet{cotler2021, Cerezo_2025} have shown that many quantum machine learning algorithms can in fact be efficiently simulated on classical hardware, a phenomenon known as \textit{dequantization} \citep{Chia_2020}. 

When the entangling Hamiltonian $H(\boldsymbol{\lambda})$ defined in Equation~\eqref{eq:add_layer} satisfies Equation~\eqref{eq:add_hamiltonian} for all parameter choices $\boldsymbol{\lambda}$ (rather than only for a particular parameter choice), the corresponding QPINN model reduces to the tensor-decomposed model proposed in Theorem \ref{theorem:tensor_poly} (without the added entangling layer). Its hypothesis function can be written in tensor-decomposed form using one-dimensional components, each of which can be obtained from a single-qubit circuit shown in Equation~\eqref{eq:u_theta_main} whose output can be written explicitly as a univariate function. Since single-qubit circuits are classically simulable with negligible cost, this tensor-decomposed model can be evaluated efficiently by computing these univariate functions and their tensor-decomposed product directly on a classical computer, without executing the underlying quantum circuit. In this case, it loses its expressivity advantage but becomes classically simulable, and should therefore be considered as a quantum-inspired PINN rather than a pure QPINN.

We emphasize that classical simulability is not necessarily a negative feature. On the contrary, it enables practical deployment on today’s hardware. If the quantum-inspired PINN shows clear theoretical or experimental advantages, it may be more practically useful than QPINNs relying on current quantum hardware.

\section{Experiments}
\label{sec:experiment}
In this section, we compare our QPINN and Quantum-inspired PINN presented in Section \ref{sec:our_QPINN} with two classical PINNs in two HJB PDE experiments for the Merton portfolio optimization problem. We first consider the HJB PDE defined in Equation~\eqref{eq:HJB_PDE}, with the parameters $r=0.02,T = 1.0, \gamma = 0.95, \mu = 0.0219, \sigma = 0.2$. These parameters define the market environment and investor preferences for solving the HJB PDE (see  Section~\ref{sec: merton portfolio optimization} for further details). \LW{We then consider the second HJB PDE defined in Equation~\eqref{eq:parametric_volatility_HJB}, where the volatility $\sigma$ is treated as an input variable.}

\subsection{Configuration}
\label{sec:exp_config}
Across both experiments, we compare four models: QPINN, Quantum-inspired PINN, Counterpart PINN, and FC PINN. The QPINN uses a tensor-decomposed model with an added entangling layer, while the Quantum-inspired PINN uses the same tensor-decomposed model without the added entangling layer. The Counterpart PINN is a classical direct parameterization that shares a similar tensor-decomposed inductive bias with the Quantum-inspired PINN. The FC PINN is a fully connected neural network with $5$ hidden layers, $10$ neurons in each hidden layer, and Tanh activations. The input dimension, tensor rank, added entangling layer, and number of trainable parameters are specified separately for each HJB PDE below.

\LW{All four models are trained with the same optimization protocol in both experiments. We use the LAMB optimizer introduced by \citet{you2020largebatchoptimizationdeep} with \texttt{weight\_decay=0} and \texttt{betas=(0.0, 0.0)}. The learning rate is initialized at $5\times 10^{-3}$ and is multiplied by $0.995$ after every optimizer step. Each training run consists of $2000$ epochs. The phase parameters of the QPINN and the Quantum-inspired PINN are initialized randomly, and the parameters of the added entangling layer are also initialized randomly when present. To reduce sensitivity to initialization, we use a multi-start selection strategy in which each model is first trained from $8$ independent candidates for $70$ epochs. We then select the candidate with the smallest total HJB loss and continue training this candidate for the remaining epochs. For simplicity, we use automatic differentiation for all derivative computations. A common output scaling factor of $5$ is applied to all models, with an additional output shift of $1$ in the second experiment. All model evaluations are carried out in PyTorch using the Apple M2 Pro CPU with the Metal Performance Shaders (MPS) backend throughout the experiments.}

\subsection{HJB PDE with fixed volatility}
\label{sec:fixed_volatility_experiment}
This experiment studies the simple case in which the analytical solution is separable in $t$ and $x$. We first specify how the four models are instantiated for this HJB PDE and then compare their training losses and learned value functions.

\subsubsection{Model configuration}
\bmhead{QPINN}
The QPINN is based on the tensor-decomposed model in Corollary \ref{coro:rank1_poly} (Theorem \ref{theorem:tensor_poly} when tensor rank $R=1$) with an added entangling layer. We take the controlled-$e^{-iH(\lambda)}$ as the added entangling layer with the 4-qubit entangling Hamiltonian $H(\lambda) = \frac{\lambda}{2}\,(I \otimes Z \otimes Z \otimes I)$ where $\lambda \in \mathbb{R}$ is a trainable parameter, as defined in Equations~\ref{eq:add_layer} and \ref{eq:add_hamiltonian}. The circuit diagram of our QPINN is shown in Figure~\ref{fig:QPINN_HJB_circuit}. The tensor-decomposed model follows Corollary \ref{coro:rank1_poly} with polynomial degree $L=1$ and number of variables $D=2$. \LW{As in the circuit block of Figure~\ref{fig:V_circuit_main}, each input variable is represented by a two-qubit register. In Figure~\ref{fig:QPINN_HJB_circuit}, qubit 1 is the Hadamard-test control, while qubits 2 and 3 form one input register and qubits 4 and 5 form the other input register.} The added entangling layer controlled-$e^{-iH(\lambda)}$ is implemented by an \texttt{IsingZZ} gate on qubits 3 and 4, with identity matrices on qubits 2 and 5, under the control of qubit 1. \LW{Thus, this \texttt{IsingZZ} gate introduces weak entanglement between the two input variables.} The QPINN therefore has 6 parameters from the tensor-decomposed model and 1 parameter from the added entangling layer, which means there are 7 trainable parameters in total. 
\begin{figure}[h]
    \centering
    \includegraphics[width=1\linewidth]{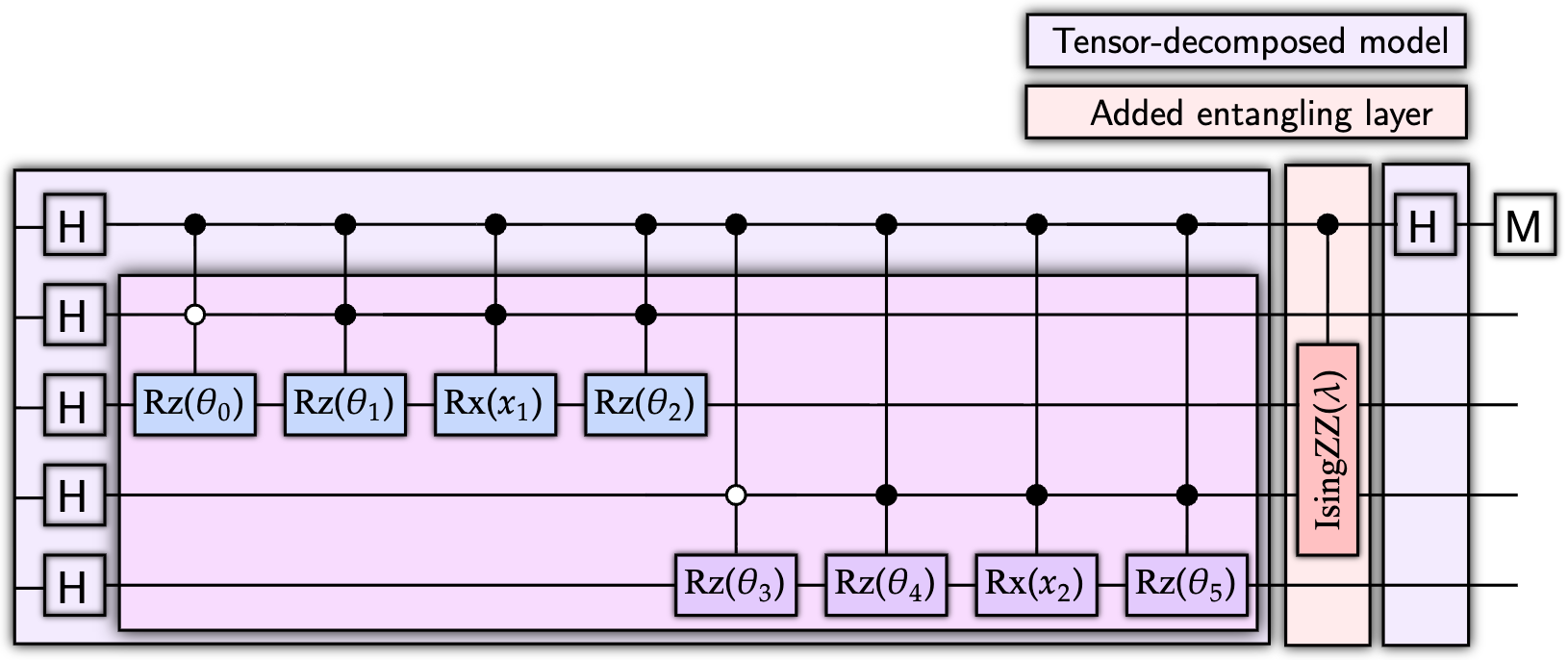}
    \caption{Circuit of QPINN for solving HJB PDE in the Merton portfolio optimization. It combines a tensor-decomposed model with an added entangling layer. The \texttt{IsingZZ} gate denotes the two-qubit operation $e^{-i \frac{\theta}{2} Z \otimes Z}$, and \texttt{Rz}, \texttt{Rx} denote single-qubit rotation gates about the Pauli-Z, Pauli-X axes, respectively. }
    \label{fig:QPINN_HJB_circuit}
\end{figure}

When $\lambda \in \mathbb{R}$ in general, we cannot express the hypothesis function explicitly, but the hypothesis space also necessarily contains functions of the tensor-decomposed structure $\{p_1(x) p_2(t)\}$. Although the added entangling layer enlarges the hypothesis space of the QPINN beyond tensor-decomposed polynomials, the part of the circuit preceding this layer still follows the tensor-decomposed structure. In fact, the actual analytical solution of this HJB PDE is $v(t, x)=\exp (-k(T-t)) \frac{x^\gamma}{\gamma}$, which is also a product of two univariate functions and is similar to the tensor-decomposed structure $p_1(x)p_2(t)$. Hence, this tensor-decomposed structure contains useful information for training and is referred to as an inductive bias. The QPINN retains a similar inductive bias to the Quantum-inspired PINN since the added entangling layer weakens it but does not eliminate it. The analytical solution is used only to interpret this benchmark, since real-world PDE problems often have no analytic solution and our proposed frameworks can also be applied to more complex cases where analytical solutions are unavailable.

\bmhead{Quantum-inspired PINN}
When the QPINN has parameter $\lambda=0$, the added entangling layer reduces to an identity matrix. The Quantum-inspired PINN is obtained by fixing $\lambda=0$ throughout training, so it uses the same tensor-decomposed model without the added entangling layer. In the present setting, the polynomial degree is $L=1$, the number of variables is $D=2$, and the tensor rank is $R=1$. Therefore, the hypothesis space contains product-form functions $p_1(x)p_2(t)$, where $p_1$ and $p_2$ are degree-$1$ univariate polynomials. This gives a rank-$1$ tensor-decomposed structure. A larger value of $L$ would enlarge the univariate polynomial factors and improve expressivity, but would also increase the required resource costs.

\bmhead{Counterpart PINN}
We compare our QPINN and Quantum-inspired PINN with their classical counterpart, the \textit{Counterpart PINN}, which shares a similar tensor-decomposed inductive bias with the Quantum-inspired PINN. The trainable parameters of the Counterpart PINN are the coefficients of two univariate polynomials $p_1$ and $p_2$ with degree $2$, giving a total of $6$ parameters.

\bmhead{FC PINN}
We also apply a commonly used fully connected PINN, the \textit{FC PINN}, which consists of $5$ hidden layers with $10$ neurons each and Tanh activations, and contains $481$ trainable parameters in total.

\subsubsection{Numerical results}
\LW{We perform $30$ independent runs for the HJB PDE with fixed volatility. The seed of run index $i=0,\ldots,29$ is set to $51+100i$. The loss function is given in Equations~\ref{eq:HJB_loss} and~\ref{eq:HJB_loss_details} with the loss weights $W_d=W_1=1$ and $W_2=5$. The residual term is evaluated on $50$ paired, uniformly spaced values of $x$ and $t$ in $[0.01,0.99]$, while the boundary terms use the corresponding grids at $t=0.99$ and $x=0.99$.} The loss value comparison of QPINN, Quantum-inspired PINN, and two PINNs for solving the Merton Portfolio Optimization HJB PDE defined in Equation~\eqref{eq:HJB_PDE} is shown in Figure \ref{fig:loss_comparison}. 

\begin{figure}[t]
    \centering
    \includegraphics[width=1\linewidth]{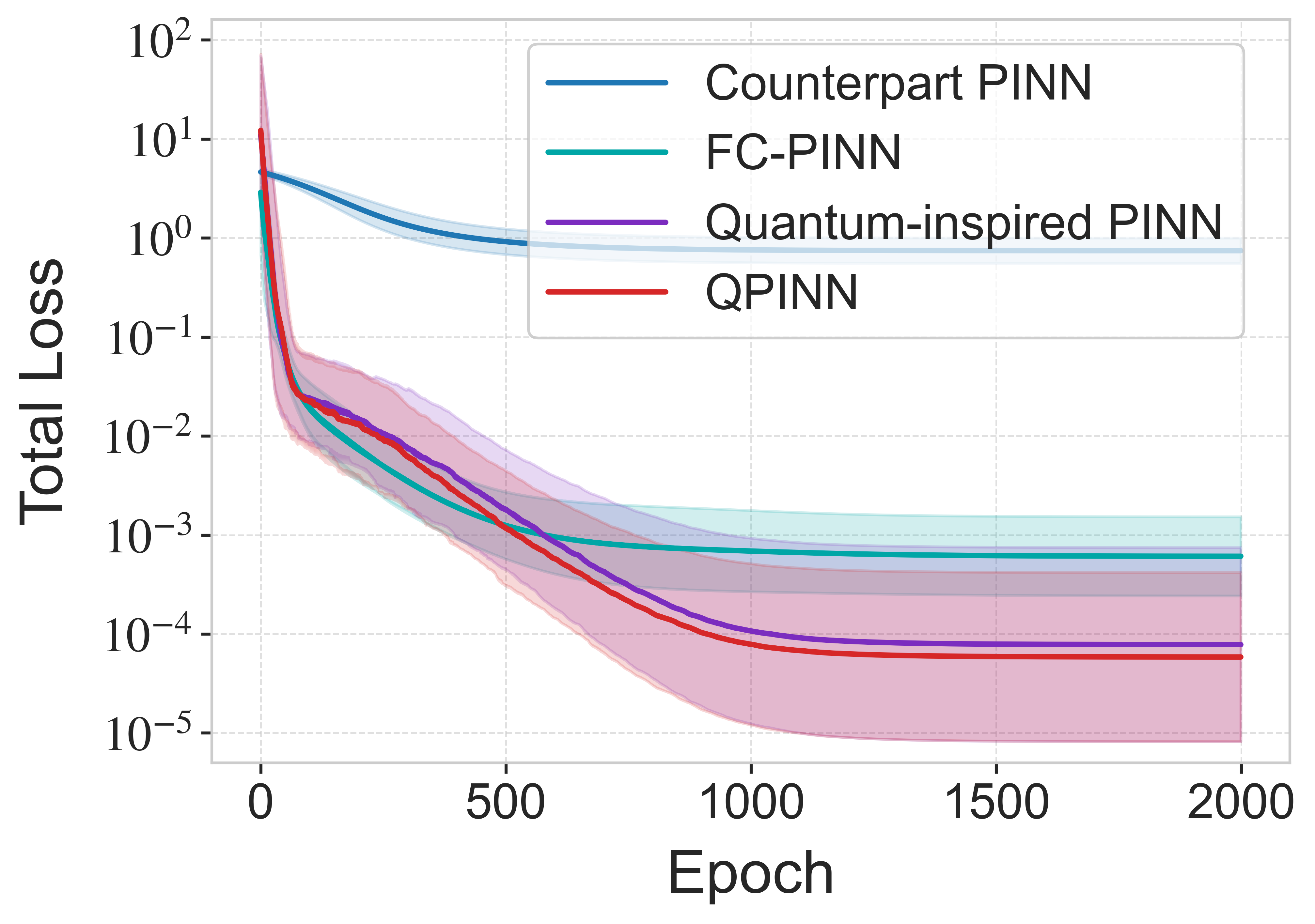}
    \caption{Comparison of loss values for solving the HJB PDE in the Merton portfolio optimization, based on 30 independent runs of the QPINN, Quantum-inspired PINN, Counterpart PINN and FC PINN, each trained for 2000 epochs. Curves represent the geometric mean loss, and shaded regions denote ±1 geometric standard deviation.}
    \label{fig:loss_comparison}
\end{figure}

\LW{In Figure \ref{fig:loss_comparison}, the QPINN attains the lowest final loss, followed by the Quantum-inspired PINN, the FC PINN, and the Counterpart PINN, even though the FC PINN has more than 80 times as many trainable parameters as the Quantum-inspired PINN. The gap between the QPINN and the Quantum-inspired PINN is visible but relatively moderate. This behavior supports the view that entanglement can enlarge the hypothesis space, while a larger empirical advantage may require a more suitable and problem-oriented design of the added entangling layer. The lower losses of the QPINN and Quantum-inspired PINN compared with the FC PINN suggest that the tensor-decomposed inductive bias is useful for this HJB PDE.}

To facilitate a more direct comparison, Figure \ref{fig:3D plot comparison} illustrates the 3D surfaces of the analytical solution and the approximations by different models obtained after 2000 training epochs in the first run. The analytical solution represents the maximum expected utility of wealth presented in Equation~\eqref{eq:analytical_solution}. The QPINN, Quantum-inspired PINN, FC PINN, and Counterpart PINN plots show how each model approximates this analytical benchmark. \LW{Visually, the QPINN and Quantum-inspired PINN surfaces are closest to the analytical solution, while the FC PINN and Counterpart PINN show larger deviations. Although no explicit boundary loss is imposed at $x=0$, the FC PINN and the QPINN/Quantum-inspired PINN recover the natural zero wealth trace $v(t, 0) \approx 0$. The annotated errors in Figure \ref{fig:3D plot comparison} follow the same order as the final losses in Figure \ref{fig:loss_comparison}: the QPINN has the smallest error, followed by the Quantum-inspired PINN, the FC PINN, and the Counterpart PINN.}
\begin{figure}[th]
    \centering
    \includegraphics[width=1\linewidth]{figures/Comparison_plot.png}
    \caption{3D comparison of the analytical solution and the approximated solutions of the HJB PDE in the Merton portfolio optimization after 2000 training epochs, based on the first run of the QPINN, Quantum-inspired PINN, Counterpart PINN, and FC PINN. The analytical solution is $v(t, x)=\exp (-0.019857375(1-t)) \frac{x^{0.95}}{0.95}$.}
    \label{fig:3D plot comparison}
\end{figure}

\begin{figure}[ht]
    \centering
    \includegraphics[width=1\linewidth]{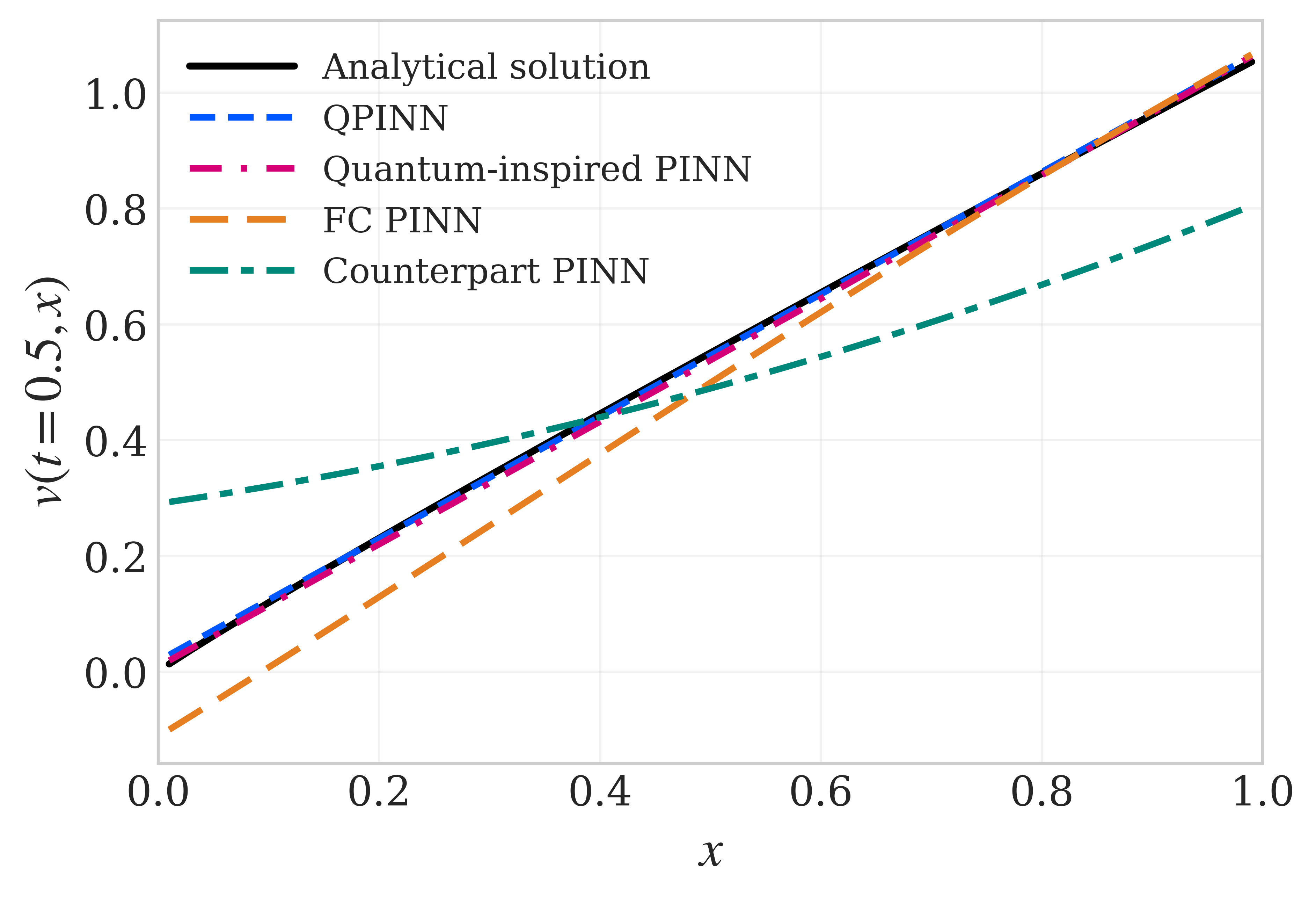}
    \caption{2D comparison of the analytical solution and the approximated solutions of the HJB PDE in the Merton portfolio optimization after 2000 training epochs, based on the first run of the QPINN, Quantum-inspired PINN, Counterpart PINN, and FC PINN at time $t=0.5$.}
    \label{fig:2D_plot_comparison}
\end{figure}

Figure \ref{fig:2D_plot_comparison} shows the fixed-time slice $v(t=0.5, x)$ for the QPINN, Quantum-inspired PINN, Counterpart PINN, and FC PINN. This view provides a clearer understanding of how each model approximates the analytical solution along the spatial dimension. The function $v(t=0.5, x)$ represents the expected utility of wealth evaluated at the mid-horizon time $t=0.5$. As observed, the QPINN approximation closely follows the analytical curve across the entire domain, achieving high accuracy both near the boundaries and within the interior region. \LW{The Quantum-inspired PINN also follows the analytical curve closely, but small deviations remain in the interior region.} In addition, the FC PINN and the Counterpart PINN exhibit even larger deviations from the analytical solution, despite the FC PINN using 481 trainable parameters, while the QPINN, Quantum-inspired PINN, and Counterpart PINN use 7, 6, and 6 trainable parameters, respectively.

\LW{Overall, these plots suggest that the tensor-decomposed structure is beneficial for this portfolio HJB PDE experiment. The QPINN gives the best loss and approximation error in the reported setting, but its advantage over the Quantum-inspired PINN is not large. This is consistent with the fact that we use only a simple representative added entangling layer, whose entangling effect is limited. The results therefore support a cautious interpretation: the added entangling layer improves performance in this case, while the main gain appears to come from using a hypothesis function matched to the separable structure of the solution.}

\subsection{HJB PDE with volatility as an input variable}
\label{sec:parametric_volatility_experiment}
\LW{This experiment uses the same four models and training protocol as before, but the input domain is enlarged from $(t,x)$ to $(t,x,\sigma)$. With a non-separable analytical solution and a higher-dimensional input domain ($D=3$), this setting is designed to evaluate whether the resource and accuracy advantages of our models also appear beyond the most favorable separable instance. We first describe the changed input structure and the added entangling layer, and then compare the training losses and learned value functions.}

\subsubsection{Model configuration}
\LW{The preceding HJB PDE experiment considered a simple benchmark whose analytical solution is separable. In that case, the input dimension is $D=2$, and our quantum models are taken with tensor rank $R=1$. We now consider a more complex benchmark with a higher-dimensional input domain and a non-separable analytical solution. The input dimension is $D=3$, and our quantum models are taken with tensor rank $R=2$. The two benchmarks also differ in the role of the added entangling layer. In the simple benchmark, the added entangling layer induces only weak entanglement with a single \texttt{IsingZZ} gate, whereas in the present benchmark it is chosen to induce much stronger entanglement with four \texttt{IsingZZ} gates to better represent this non-separable structure.}

\LW{We then apply the four models to the HJB PDE defined in Equation~\eqref{eq:parametric_volatility_HJB}, treating the volatility $\sigma\in[0.15,0.25]$ as an input variable. The analytical solution in Equation~\eqref{eq:parametric_volatility_solution} contains the factor $\exp(-k(\sigma)(T-t))$. Since $k(\sigma)$ depends on $\sigma$, the variables $t$ and $\sigma$ enter this factor jointly, and the solution is not separable with respect to these two variables. This benchmark therefore no longer provides the purely separable structure that most directly matches the tensor-decomposed inductive bias.}

\LW{By Theorem~\ref{theorem:tensor_poly} and the circuit block shown in Figure~\ref{fig:V_circuit_main}, the tensor-decomposed model represents each input variable by a two-qubit register. In the $D=3$ experiment with input variables $x,t,\sigma$, we denote these registers by $\mathcal{S}_x$, $\mathcal{S}_t$, and $\mathcal{S}_\sigma$, so the three input variables are encoded on six variable qubits in total.   This allocation reflects the different levels of expressivity needed to approximate the dependence of the analytical solution along each variable. The analytical solution in Equation~\eqref{eq:parametric_volatility_solution} factors into $x^\gamma/\gamma$ and $\exp(-k(\sigma)(T-t))$. Hence the dependence on $x$ remains separable from the pair $(t,\sigma)$, while $t$ and $\sigma$ are coupled through $k(\sigma)(T-t)$. For this reason, the added entangling layer leaves the $x$ register unchanged and introduces entanglement only between the two-qubit $t$ register and the two-qubit $\sigma$ register. This added entanglement provides extra expressivity for representing the non-separable dependence between $t$ and $\sigma$. To reduce computational cost, the added entangling layer parameters are optimized only from epoch $500$ onward, since the early training stage is less likely to be limited by expressivity. We label the Pauli-$Z$ operators on the two qubits of $\mathcal{S}_t$ by $Z_1,Z_2$, and the Pauli-$Z$ operators on the two qubits of $\mathcal{S}_\sigma$ by $Z_3,Z_4$. The added entangling layer on the six variable qubits is $I_x\otimes e^{-iH_{t\sigma}(\boldsymbol{\lambda})}$, where the entangling Hamiltonian is
\begin{equation}
\label{eq:volatility_input_added_hamiltonian}
\begin{aligned}
H_{t\sigma}(\boldsymbol{\lambda})
&=\frac{1}{2}\left(
\lambda_{13} Z_1 Z_3
+\lambda_{14} Z_1 Z_4 \right.\\
&\quad\left.
+\lambda_{23} Z_2 Z_3
+\lambda_{24} Z_2 Z_4
\right).
\end{aligned}
\end{equation}
where $\boldsymbol{\lambda}=(\lambda_{13},\lambda_{14},\lambda_{23},\lambda_{24})\in\mathbb{R}^4$ are trainable parameters. Here $I_x$ is the identity on the two-qubit $x$ register. The four $Z_iZ_j$ terms couple each $t$-qubit with each $\sigma$-qubit, and therefore define a direct entangling layer between $\mathcal{S}_t$ and $\mathcal{S}_\sigma$. When $\boldsymbol{\lambda}=\boldsymbol{0}$, this layer reduces to the identity, so the construction follows Equations~\ref{eq:add_layer} and \ref{eq:add_hamiltonian}. The complete circuit and its register-level interpretation are provided in Appendix~\ref{sec:volatility_input_qpinn_circuit}.}

\LW{For the QPINN and the Quantum-inspired PINN, the tensor-decomposed model has tensor rank $R=2$ over the three input variables $(t,x,\sigma)$. The QPINN uses this tensor-decomposed model together with the added entangling layer defined in Equation~\eqref{eq:volatility_input_added_hamiltonian}. In this setting, the Quantum-inspired PINN and the Counterpart PINN each contain $30$ trainable parameters, while the QPINN contains $34$ trainable parameters because of the additional entangling parameters. The FC PINN has the same fully connected structure as before, with input dimension $3$, and contains $491$ trainable parameters.}

\subsubsection{Numerical results}
\LW{Figure~\ref{fig:loss_comparison_volatility_input} shows the loss curves for this experiment with $\sigma$ as an input variable. The QPINN obtains the lowest final loss, followed by the Quantum-inspired PINN and the FC PINN, while the Counterpart PINN remains at a significantly larger loss. Compared with the fixed-volatility experiment in Figure~\ref{fig:loss_comparison}, the gap between the QPINN and the Quantum-inspired PINN becomes larger. This is consistent with our design choice: the added entangling layer is stronger in this experiment and is introduced to increase expressivity for the non-separable dependence between $t$ and $\sigma$. At the same time, the Quantum-inspired PINN still outperforms the FC PINN because the solution remains separable between $x$ and the pair $(t,\sigma)$, and the tensor rank is increased to $R=2$. The parameter gap relative to the FC PINN is nevertheless much smaller than in the fixed-volatility experiment, decreasing from about $80$ times to about $15$ times, which reflects the reduced advantage of the tensor-decomposed inductive bias in this less favorable setting. These results show that the resource and accuracy advantages also appear beyond the most favorable separable instance.}

\begin{figure}[t]
    \centering
    \includegraphics[width=1\linewidth]{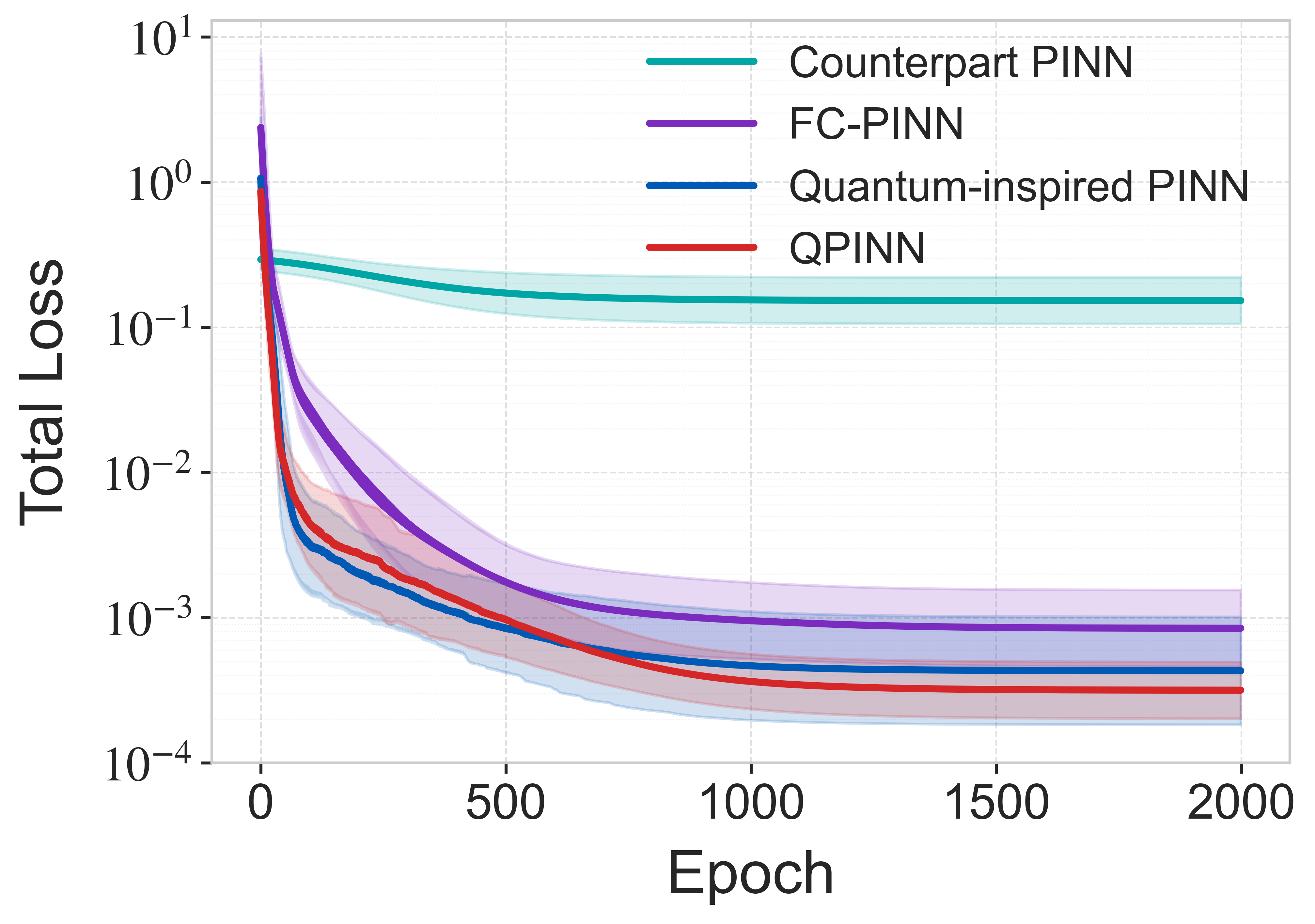}
    \caption{Comparison of loss values for solving the HJB PDE in Equation~\eqref{eq:parametric_volatility_HJB}, with volatility $\sigma$ treated as an input variable. The results are based on $10$ independent runs of the QPINN, Quantum-inspired PINN, Counterpart PINN, and FC PINN, each trained for $2000$ epochs. Curves represent the geometric mean loss, and shaded regions denote $\pm 1$ geometric standard deviation.}
    \label{fig:loss_comparison_volatility_input}
\end{figure}

\LW{To compare the learned value functions directly, Figure~\ref{fig:3D_comparison_volatility_input} reports the analytical solution and the four learned approximations after fixing $\sigma=0.15$ for visualization. The learned approximations are functions of the full three-dimensional input $(t,x,\sigma)$, and the figure shows their two-dimensional slices over $(t,x)$ after $2000$ training epochs in the first run. The relative errors shown in the figure follow the same qualitative order as the loss curves: the QPINN is closest to the analytical solution, followed by the Quantum-inspired PINN and the FC PINN, while the Counterpart PINN gives the largest deviation. These results are consistent with the interpretation that the added \texttt{IsingZZ} entangling layer helps the QPINN capture the coupled dependence on $(t,\sigma)$ beyond the separable structure used by the tensor-decomposed inductive bias.}

\begin{figure}[h]
    \centering
    \includegraphics[width=1\linewidth]{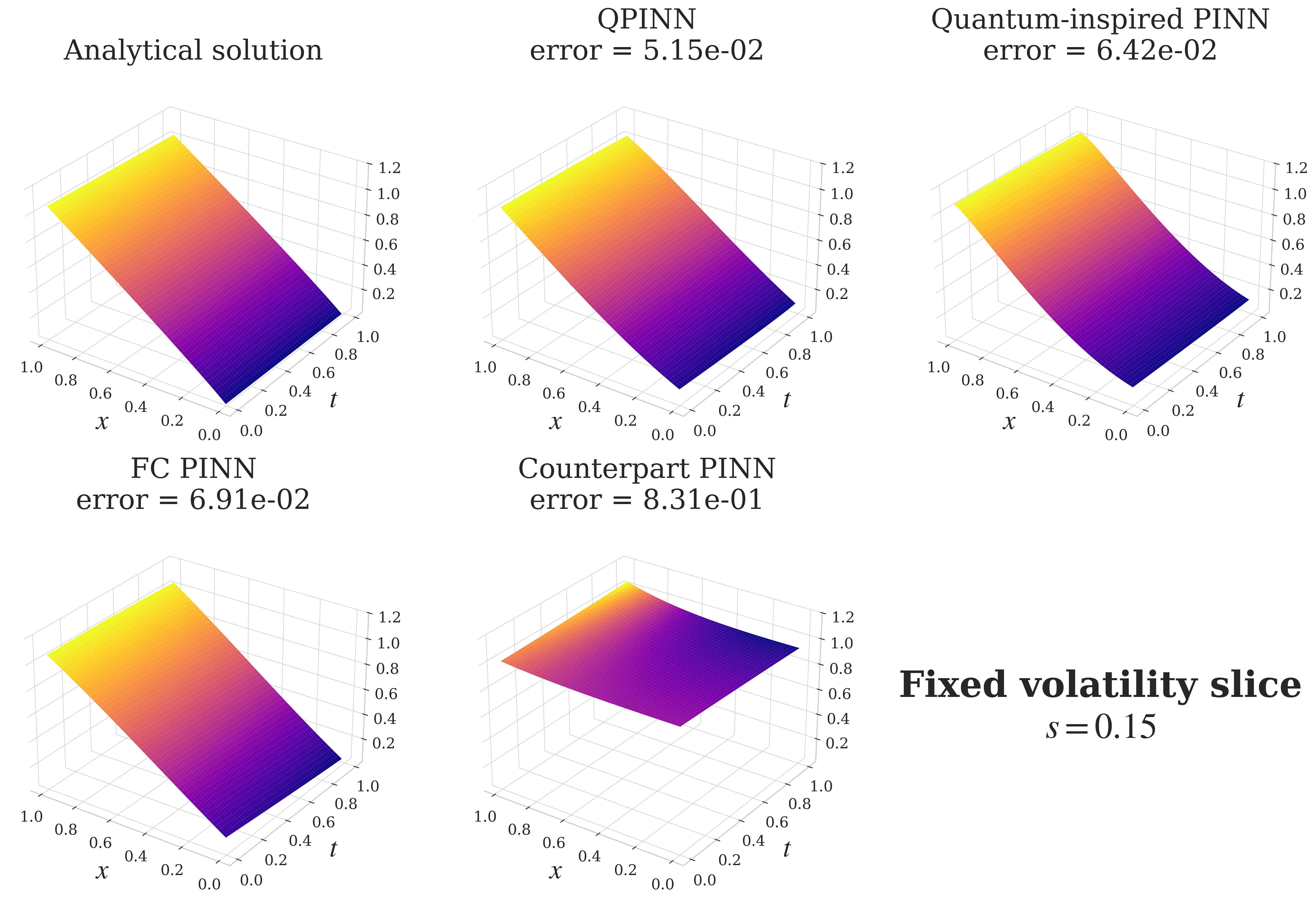}
    \caption{3D comparison of the analytical solution and the approximated solutions of the HJB PDE in Equation~\eqref{eq:parametric_volatility_HJB} after $2000$ training epochs, based on the first run. The learned approximations are functions of $(t,x,\sigma)$, and the figure shows the fixed-volatility slice $\sigma=0.15$. The annotated errors are relative $\ell_2$ errors against the analytical solution on the plotted grid.}
    \label{fig:3D_comparison_volatility_input}
\end{figure}

\subsection{Financial interpretation}
To illustrate the financial interpretation, we focus on the fixed-volatility HJB PDE in Equation~\eqref{eq:HJB_PDE}, with parameters $r=0.02$, $T=1.0$, $\gamma=0.95$, $\mu=0.0219$, and $\sigma=0.2$. Here, $r$ denotes the risk-free interest rate, $T$ is the investment horizon, and $\mu$ and $\sigma$ represent the expected return and volatility of the risky asset, respectively. The parameter $\gamma$ characterizes the investor's attitude toward risk. In this setting, the analytical solution to the HJB PDE is $v(t, x)=\exp (-0.019857375(1-t)) \frac{x^{0.95}}{0.95}$, which represents the maximum expected utility of wealth. The optimal investment fraction in the risky asset is constant and equal to $\hat{\alpha}=0.95$, meaning that 95\% of the portfolio should be allocated to the risky asset and the remaining 5\% to the risk-free asset. The QPINN and PINN models can be used to approximate the value function $v(t, x)$. Once the hypothesis functions are obtained from the approximation shown in Figure \ref{fig:3D plot comparison}, the optimal control $\hat{\alpha}(t)$ can be directly derived from the HJB PDE through Equations~\ref{eq:init_HJB} and \ref{eq:general_optimal_control}, allowing an investor to achieve the maximum expected wealth. 

\section{Conclusion}
\label{sec:conclusion}
In this work, we have introduced a family of quantum models capable of implementing different forms of polynomials, including their univariate, multivariate, and tensor-decomposed variants for solving PDEs in the framework of the PINN model. It is noted that our results reduce the required quantum circuit resources from exponential to polynomial complexity for implementing tensor-decomposed polynomials when the tensor rank is not too large. Hence, it can be seen as an extension to the results presented in \citet{yu2024nonasymptotic}, although the underlying assumptions and complexity are different (see Table \ref{table:comparison}). \LW{From a broader perspective, our results indicate that tensor-decomposed models are useful for structured PDE settings where the solution has an exact or low-rank decomposable form, rather than for PDEs in general.} \LW{The two HJB PDE experiments further show that this advantage is not limited to the most favorable separable instance: when volatility is treated as an input variable, the Quantum-inspired PINN still outperforms the FC PINN with fewer parameters, although the parameter gap becomes smaller.} We further find that such quantum tensor-decomposed models can be simulated efficiently on classical hardware, which implies that \LW{the tensor-decomposed inductive bias can be realized within PINN architectures implemented entirely on classical computers.}

\LW{In contrast, the added entangling layer yields a QPINN that is no longer efficiently simulable on classical hardware.} While many quantum machine learning approaches offer limited guarantees regarding the existence of suitable solutions within their hypothesis spaces, our QPINN retains the potential expressivity benefits of quantum entanglement while ensuring that the hypothesis space contains at least one suitable solution. Although our work does not provide an explicit error bound for this approximation, such bounds can be inferred from classical results on polynomial approximations of (continuous) target functions. \LW{The fixed-volatility experiment uses a simple added entangling layer and shows only a moderate gap between the QPINN and the Quantum-inspired PINN. In the volatility-input experiment, the QPINN uses a stronger added entangling layer, and the larger gap between the QPINN and the Quantum-inspired PINN indicates that a problem-oriented entangling design can provide useful additional expressivity for non-separable dependence.}

Our study primarily addresses expressivity, leaving the equally crucial question of theoretical trainability open for future work. Since trainability is widely believed to be an important source of potential quantum neural network advantages \citep{Schreiber2023Classical}, future work should integrate expressivity analysis with a systematic investigation of optimization landscapes and gradient behavior. Moreover, a rigorous quantification of approximation errors between computational model hypothesis functions and PDE solutions should be essential to establish stronger theoretical guarantees.

Overall, our work develops both a quantum resource analysis for implementing tensor-decomposed polynomials and a hypothesis space characterization of the QPINN and the Quantum-inspired PINN. For both the QPINN and the Quantum-inspired PINN, the hypothesis space contains the entire family of tensor-decomposed polynomials, which guarantees the existence of an approximation to the PDE solution. In the case of the QPINN, the added entangling layer further expands the hypothesis space beyond tensor-decomposed polynomials. \LW{Through a suitable choice of the entangling Hamiltonian, it can also introduce a problem-oriented inductive bias, providing additional expressivity that the Quantum-inspired PINN does not possess. Moreover, because these architectures are built upon an efficient tensor-decomposed structure, they may provide a useful framework for exploring structured PDE solvers on near-term quantum hardware.}
\backmatter


\bmhead{Author Contributions}
L.W. conceived the research idea, developed the theoretical framework, implemented the models, performed the experiments, and wrote the manuscript.
A.L., S.S., and Z.T. supervised the research, provided guidance throughout the project, and contributed to the discussion, revision, and improvement of the manuscript.
All authors reviewed and approved the final version of the paper.

\bmhead{Funding}
This research was supported by French government under the France 2030 program, reference ANR-11-IDEX-0003 within the OI H-Code.

\bmhead{Data Availability}
The source code implementation is available at [\texttt{https://github.com/Letao-WANG
/QPINNs-for-Portfolio/tree/main}].

\section*{Declarations}

\textbf{Conflicts of interest} The authors declare no competing interests.

\bibliography{sn-bibliography}


\clearpage
\onecolumn
\begin{appendices}

\section{Tensor-decomposed polynomial derivation}
\label{sec: TD_polynomials}
We first recall the notion of \textit{tensor rank decomposition} \citep{kolda2009tensor}. Given an index $n_m\in I_m,\forall m\in[D]$ with $D$ as an integer, let $\mathcal A\in\mathbb F^{I_1\times I_2\times\cdots\times I_D}$ be a $D$-order tensor over a field~$\mathbb F$ with entries $\mathcal A_{n_1,\dots,n_D}$. A tensor rank decomposition is a representation
\begin{equation}
\label{eq:tensor}
\mathcal A
=\sum_{r=1}^R
\lambda_r\,
\bigl(
\mathbf a_r^{(1)}\otimes
\mathbf a_r^{(2)}\otimes
\cdots\otimes
\mathbf a_r^{(D)}
\bigr),
\end{equation}
where $\lambda_r\in\mathbb F $ is a coefficient, $\mathbf a_{r}^{(m)}\in\mathbb F^{I_m}$ is a tensor factor, $R$ denotes the decomposition rank, and $\otimes$ denotes the outer product. For any $D'$-order tensor  $\mathbf{a}\in\mathbb F^{I_1\cdots\times I_{D'}}$, we can consider $\mathbf{a}_{n_1,\dots,n_{D'}}\in\mathbb F$ as the $\left(n_1, \ldots, n_{D'}\right)$ entry of $\mathbf{a}$. By using the entry equation
$$
\bigl(\mathbf a_r^{(1)}\otimes
\mathbf a_r^{(2)}\otimes
\cdots\otimes
\mathbf a_r^{(D)}\bigr)_{n_1,\dots,n_D}
=
\mathbf a_{r,n_1}^{(1)}
\mathbf a_{r,n_2}^{(2)}
\cdots
\mathbf a_{r,n_D}^{(D)}
\equiv
\mathbf a_{r,n_1}
\mathbf a_{r,n_2}
\cdots
\mathbf a_{r,n_D},
$$
we can obtain the entry of tensor $\mathcal A$:
\begin{equation}
\label{eq:tensor_entry}
\mathcal A_{n_1,\dots,n_D}=\sum_{r=1}^R \lambda_r\mathbf{a}_{ r,n_1}\mathbf{a}_{ r,n_2} \cdots \mathbf{a}_{ r,n_D},
\end{equation}
where $\mathcal A_{n_1,\dots,n_D}\in\mathbb F,\lambda_r\in\mathbb F$ and $\mathbf{a}_{ r,n_j}\in\mathbb F,\forall j\in[D]$. It is worth noting that there is a variant of the tensor rank decomposition, known as the \textit{CP decomposition} \citep{Hitchcock1927TheEO}. The key difference is that tensor rank decomposition provides an exact decomposition, whereas CP decomposition yields an approximate one, with the rank specified by the user. From Equation \eqref{eq:tensor_entry}, the multivariate polynomial coefficients $c_{\boldsymbol{n}}$ in Definition \ref{def:multivariate polynomial} can be expressed as
\begin{equation}
\label{eq:tensor_general}
c_{\boldsymbol{n}}=c_{n_1,\dots,n_D}
=\sum_{r=1}^R
\lambda_r c_{r,n_1}\;c_{r,n_2}\;\cdots\;c_{r,n_D}.
\end{equation}
Then we define another version of a univariate polynomial with degree $L$ such that
\begin{equation}
\label{eq:another_poly}
p_{r, j}\left(x_j\right)=\sum_{n_j}c_{r,n_j}x_j^{n_j},
\end{equation}
where $c_{r,n_j}\in\mathbb{R},n_j\in[L],\forall r\in[R],j\in[D]$. Combining Definition \ref{def:multivariate polynomial} for real multivariate polynomials with Equations \eqref{eq:tensor_general} and \eqref{eq:another_poly}, and assuming $\lambda_r\in\mathbb{R},\forall r\in[R]$, we obtain
\begin{align*}
\sum_{r=1}^R \lambda_r\prod_{j=1}^{D}p_{r,j}(x_j)
=\sum_{r=1}^R \lambda_r\prod_{j=1}^{D}\sum_{n_j}c_{r,n_j}x_j^{n_j}
=\sum_{r=1}^R \lambda_r\sum_{n_1} \sum_{n_2} \cdots \sum_{n_D}\prod_{j=1}^{D}( c_{r,n_j} x_j^{n_j})\\
=\sum_{n_1} \sum_{n_2} \cdots \sum_{n_D}\left(\sum_{r=1}^R \lambda_rc_{r,n_1} c_{r,n_2} \cdots c_{r,n_D}\right) x_1^{n_1}x_2^{n_2}\cdots x_D^{n_D}
=\sum_{\boldsymbol{n}}c_{\boldsymbol{n}} x^{\boldsymbol{n}}.
\end{align*}

It should be noted that for any real multivariate polynomial $p(\boldsymbol{x})$, there always exists an integer $R$ such that its coefficients admit Equation~\eqref{eq:tensor_general}. In the worst case, one can represent each monomial 
of $p(\boldsymbol{x})$ by univariate polynomials $p_{r,j}(x_j)$, so that $R$ is bounded above by the number of possible multi-indices $\boldsymbol{n}\in[L]^D$. Hence, $1\leq R \leq (L+1)^D$. 

We then give the formal definition: let $p(\boldsymbol{x})$ be a real multivariate polynomial with $D$ variables and degree at most $L$ in each variable in Definition \ref{def:multivariate polynomial}. We can always find coefficients $(c_{\boldsymbol{n}})$ that admit a decomposition of Equation \eqref{eq:tensor_general} with rank $R$, so $p(\boldsymbol{x})$ can be written as
\begin{equation*}
p(\boldsymbol{x})=\sum_{r=1}^R \lambda_r \prod_{j=1}^D p_{r,j}(x_j),
\end{equation*}
where each $p_{r,j}(x_j)$ is a univariate polynomial, $\lambda_r\in\mathbb{R},c_{r,n_j}\in\mathbb{R},n_j\in[L],\forall r\in[R],j\in[D]$. We call such $p(\boldsymbol{x})$ a tensor-decomposed polynomial (TD polynomial).

\section{Merton portfolio optimization derivation}
\label{sec:appendix Merton portfolio optimization}
In this work, we consider only the case without a jump component presented in \citet{jump_diff}. In particular, let $S_0=\left(S_0(t)\right)_{t \in[0, T]}$ and $S_1=\left(S_1(t)\right)_{t \in[0, T]}$ denote the amount of money the investor has in the risk-free asset and risky asset, respectively. The processes evolve according to
\begin{equation}
\label{eq:basic_def}
\left\{\begin{array}{l}
d S_0(t)=r S_0(t) d t \\
d S_1(t)=S_1(t)(\mu(t) d t+\sigma(t) d B(t)) \\
S_0(0)=1, S_1(0)=s
\end{array}\right.
\end{equation}
where $B(t)$ is a standard Brownian motion, $\mu(t)$ is the expected return of the asset at time $t$, and $\sigma(t)$ is the volatility of the asset at time $t$. We assume that $\mu(t)$ and $\sigma(t)$ are both stochastic processes.

Let $\pi=\left(\pi_0(t), \pi_1(t)\right)_{t \in[0, T]}$ denote an investor portfolio where $\pi_0(t)$ and $\pi_1(t)$ represent the proportion of wealth invested in the risk-free asset and in the risky asset, respectively. Therefore, for every $t \in[0, T]$, they satisfy the relation $\pi_0(t)+\pi_1(t)=1$. The wealth at time $t$ of such a portfolio is
\begin{equation}
\label{eq:X(t)}
X(t)=\pi_0(t) S_0(t)+\pi_1(t) S_1(t),
\end{equation}
and represents the total amount of money invested in the market. The portfolio is assumed to be self-financing, that is,
\begin{equation}
\label{eq:d X(t)}
d X(t)=\pi_0(t) d S_0(t)+\pi_1(t) d S_1(t).
\end{equation}
Combining Equation \eqref{eq:basic_def} \eqref{eq:X(t)} \eqref{eq:d X(t)}, the investor's wealth process is given by
\begin{equation}
\label{eq:pre_SDE}
d X(t)=X(t) r d t+\pi_1(t) S_1(t)(\mu(t)-r) d t+\pi_1(t) S_1(t) \sigma(t) d B(t).
\end{equation}
We define $X^{t, x, \alpha}(s)$ as the expected wealth according to Equation \eqref{eq:pre_SDE} starting from $X(s=t)=x$ and evolving until $s=T$, we also have
\begin{equation}
\label{eq:pre_SDE_evolve}
d X^{t, x, \alpha}(s)=X^{t, x, \alpha}(s) r d s+\pi_1(s) S_1(s)(\mu(s)-r) d s+\pi_1(s) S_1(s) \sigma(s) d B(s).
\end{equation}
Let $\alpha(t)$ denote the fraction of the total wealth invested in stocks at time $t$ such that
\begin{equation}
\label{eq:alpha(t)}
\alpha(t)=\frac{\pi_1(t) S_1(t)}{X(t)} \Longleftrightarrow \pi_1(t) S_1(t)=\alpha(t) X(t).
\end{equation}
Using Equation \eqref{eq:alpha(t)} \eqref{eq:pre_SDE_evolve}, one obtains
\begin{equation}
d X^{t, x, \alpha}(s)=X^{t, x, \alpha}(s)[(\alpha(s)(\mu(s)-r)+r) d s+\alpha(s) \sigma(s) d B(s)].
\end{equation}
The investor's objective is to maximize the expected utility over $[0, T]$, which is defined as
$$
\mathcal{J}(t, x, \alpha)=\mathbb{E}\left[U\left(X^{t, x, \alpha}(T)\right)\right],
$$
where $U$ is the investor's utility function. We introduce the following assumptions on the utility function:
\begin{itemize}
    \item $U(x)$ is a continuous, non-decreasing, and concave function on $[0, \infty)$ with $U(0)=0$.
    \item There exists $\gamma > 0$ and a constant $K>0$ such that $U(x) \leq K(1+x)^\gamma$ for all $x \in[0, \infty)$.
\end{itemize}
The Merton portfolio optimization problem is to find a function $v$ and an optimal control $\hat{\alpha}$ such that
\begin{equation}
\label{eq:optimization}
v(t, x)=\sup _{\alpha \in \mathcal{A}} \mathcal{J}(t, x, \alpha)=\mathcal{J}(t, x, \hat{\alpha}), \quad(t, x) \in[0, T] \times[0,+\infty),
\end{equation}
which is an optimal control problem. In fact, the function $v$ also solves its corresponding Hamilton-Jacobi-Bellman (HJB) equation:
\begin{equation}
\label{eq:init_HJB}
\begin{cases}\frac{\partial v}{\partial t}(t, x)+H\left(t, x, \frac{\partial v}{\partial x}(t, x), \frac{\partial^2 v}{\partial x^2}(t, x)\right)=0, & (t, x) \in[0, T) \times[0,+\infty) \\ v(T, x)=U(x), & x \in[0,+\infty)\end{cases}
\end{equation}
where, for $H(t, x, p, q) \in[0, T] \times \mathbb{R} \times \mathbb{R} \times \mathbb{R}$,
$$
H(t, x, p, q)=\sup _{\alpha \in A}\left\{p x r+p x(\mu-r) \alpha+\frac{1}{2} x^2 \sigma^2 q \alpha^2\right\},
$$
is the Hamiltonian of the problem. 

Let $\psi(\alpha):=p x r+p x(\mu-r) \alpha+\frac{1}{2} x^2 \sigma^2 q \alpha^2$. Assuming $\mu(t)$ and $\sigma(t)$ as constant,  $\mu > r$ and $U(x)=\frac{x^\gamma}{\gamma}$ with risk $\gamma$, we have
$$
\psi^{\prime}(\alpha)=p x(\mu-r)+q \sigma^2 x^2 \alpha.
$$
Let $\hat{\alpha}$ be such that $\psi^{\prime}(\hat{\alpha})=0$, which means
\begin{equation}
\label{eq:general_optimal_control}
\hat{\alpha}=-\frac{\mu-r}{\sigma^2} \frac{p}{q x}.
\end{equation}
When the investor is risk-averse, i.e., $\gamma\in(0,1)$, the point $\hat{\alpha}$ corresponds to the maximum of $\psi(\alpha)$, so we obtain
\begin{equation}
\label{eq:Hamilton}
H(t, x, p, q)=p r x-\frac{1}{2} \frac{p^2}{q}\left(\frac{\mu-r}{\sigma}\right)^2.
\end{equation}
The derivation of Equation \eqref{eq:Hamilton} is presented in \citet{jump_diff}. Hence, we have the corresponding HJB PDE from Equation \eqref{eq:init_HJB} and \eqref{eq:Hamilton} on the domain $(t, x) \in[0, T] \times[0,+\infty)$: 
\begin{equation}
\label{eq:PDE_appendix}
\begin{cases}
\partial_t v(t, x)\partial^2_x v(t, x)+\partial_x v(t, x) \partial^2_x v(t, x) r x-\frac{1}{2} \left(\frac{\mu-r}{\sigma}\right)^2 \partial_x v(t, x)^2=0,\\[1pt]
v(T,x)=\frac{x^\gamma}{\gamma},\\[1pt]
v(t,1)=\exp (-k(T-t))\frac{1}{\gamma}, 
\end{cases}
\end{equation}
with the analytical solution $v(t, x)=\exp (-k(T-t)) \frac{x^\gamma}{\gamma}$ where $k=\frac{1}{2} \frac{\gamma}{\gamma-1}\left(\frac{\mu-r}{\sigma}\right)^2-r \gamma$, which can be interpreted as the maximum expected utility of wealth. The derivation of the analytical solution and the optimal control $\hat{\alpha}=\frac{1}{1-\gamma} \frac{\mu-r}{\sigma^2}$ is also shown in \citet{jump_diff}. The solution $v(t,x)$ of the HJB PDE also solves Equation \eqref{eq:optimization}. We can then apply QPINNs to solve the HJB PDE, then calculate the optimal control based on Equation \eqref{eq:general_optimal_control}, and thus obtain the Merton portfolio problem solution. 

Note that the simplifying assumptions made here are just for illustration; in general, solving the associated HJB PDE is a standard approach for addressing optimal control problems, which can then be solved numerically using QPINNs. These two boundary conditions $v(T,x)=\frac{x^\gamma}{\gamma},v(t,1)=\exp (-k(T-t))\frac{1}{\gamma}$ in Equation \eqref{eq:PDE_appendix} are not unique, and other admissible choices can equally be considered, for example $v(t,0)=0,t\in[T]$.

\section{Implementing polynomials derivation}

\subsection{Polynomial implementation}

In this subsection, we present existing results from the QSP framework and explain how they are used to implement polynomial functions within quantum models, along with an analysis of the corresponding quantum resource complexity.

\subsubsection{Univariate case}
\label{sec:univariat_poly}
We first introduce the quantum model for implementing real univariate polynomials, based on QSP lemmas from \citet{Gily_n_2019}. Define the encoding operator $S(x)$ such that
$$
S(x):=R_x(-2 \arccos (x))=\left(\begin{array}{cc}x & i \sqrt{1-x^2} \\ i \sqrt{1-x^2} & x\end{array}\right),
$$
where $x\in[-1,1]$ and
\begin{equation}
    \label{eq:u_theta}
    U_{\boldsymbol{\theta}}(x):=R_z\left(\theta_L\right) \prod_{j=0}^{L-1} S(x) R_z\left(\theta_j\right).
\end{equation}
\begin{lemma}[\citealp{Gily_n_2019}]
\label{lem:init_poly}
There exists $\boldsymbol{\theta} \in \mathbb{R}^{L+1}$ such that
$$
p(x)=\langle+| U_{\boldsymbol{\theta}}(x)|+\rangle
$$
if and only if the real polynomial $p(x) \in \mathbb{R}[x]$ satisfies
\begin{enumerate}
    \item  $\operatorname{deg}(p(x)) \leq L$,
    \item $p(x)$ has parity $L \bmod 2$ (i.e. if $L$ is even/odd then $p(x)$ must be an even/odd function),
    \item $\forall x \in[-1,1],|p(x)| \leq 1$.
\end{enumerate}
\end{lemma}

The circuit diagram of $U_{\boldsymbol{\theta}}(x)$ is shown in Figure \ref{fig:single_qubit_QSP_main}. We have several different methods to estimate $\langle+| U_{\boldsymbol{\theta}}(x)|+\rangle$ to obtain $p(x)$. In this work, we consider only the Hadamard test method, whose circuit diagram is shown in Figure \ref{fig:Hadamard_test}.
\begin{figure}[H]
    \centering
    \includegraphics[width=0.3\linewidth]{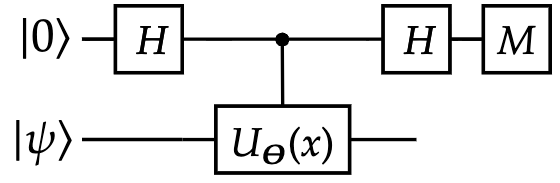}
    \caption{Hadamard test to estimate the real part of $\langle \psi| U_{\boldsymbol{\theta}}(x)|\psi\rangle$.}
    \label{fig:Hadamard_test}
\end{figure}
For the Hadamard test method, regardless of the estimated term, it is necessary to measure the first qubit of the circuit $O\left(1 / \varepsilon^2\right)$ times in order to estimate the expected value within an error bound $\varepsilon$. In general, $\langle\psi| U|\psi\rangle$ is complex for an arbitrary unitary $U$ and quantum state $|\psi\rangle$. The circuit shown in Figure \ref{fig:Hadamard_test} can only access the real part of $\langle\psi| U|\psi\rangle$, and therefore cannot recover the full complex value. However, when $\langle+| U_\theta(x)|+\rangle$ is real, it allows us to use the circuit in Figure \ref{fig:Hadamard_test} to exactly recover $p(x)$. The same idea applies to other results.

Although Lemma \ref{lem:init_poly} relies on a counter-intuitive parity constraint, this condition can be relaxed. In the following theorem, we strengthen Lemma \ref{lem:init_poly} by removing its parity constraint—at the cost of introducing a $1/2$ normalizing factor, which in practice causes no significant difficulties.

\begin{corollary}
\label{coro:univariate_poly_coro}
There exists $\boldsymbol{\theta}_1 \in \mathbb{R}^{L},\boldsymbol{\theta}_2 \in \mathbb{R}^{L+1}$ such that
$$
p(x)=\langle+|^{\otimes2} (|0\rangle\langle0|\otimes U_{\boldsymbol{\theta_{1}}}(x)+|1\rangle\langle1|\otimes U_{\boldsymbol{\theta_{2}}}(x))|+\rangle^{\otimes2},
$$
where $U_{\boldsymbol{\theta_{1}}}(x),U_{\boldsymbol{\theta_{2}}}(x)$ are defined as Equation \eqref{eq:u_theta} if the real polynomial $p(x) \in \mathbb{R}[x]$ satisfies
\begin{enumerate}
    \item  $\operatorname{deg}(p(x)) \leq L$,
    \item $\forall x \in[-1,1],|p(x)| \leq \frac{1}{2}$.
\end{enumerate}
\end{corollary}
\begin{proof}
Let $p(x)=\frac{1}{2}p_{\text{odd}}(x)+\frac{1}{2}p_{\text{even}}(x)$  where $p_{\text {even }}(x)=p(x)+p(-x)$ and $p_{\text {odd }}(x)=p(x)-p(-x)$. We have $\deg(p_{\text {even }}(x))\leq L,\deg(p_{\text {odd}}(x))\leq L-1$ when $L$ is even, and $\deg(p_{\text {even }}(x))\leq L-1,\deg(p_{\text {odd}}(x))\leq L$ when $L$ is odd. We also know that $|p_{\text {odd}}(x)|\leq 1$ and $|p_{\text {even}}(x)|\leq 1$ since $|p(x)|\leq\frac{1}{2}$. Hence from the Lemma \ref{lem:init_poly}, we know there exist  $\boldsymbol{\theta}_{\text{odd}}\in \mathbb{R}^{L},\boldsymbol{\theta}_{\text{even}} \in \mathbb{R}^{L+1}$ ($L$ is even) or $\boldsymbol{\theta}_{\text{odd}}\in \mathbb{R}^{L+1},\boldsymbol{\theta}_{\text{even}} \in \mathbb{R}^{L}$ ($L$ is odd) such that 
$$
p_{\text{odd}}(x)=\langle+| U_{\boldsymbol{\theta_{\text{odd}}}}(x)|+\rangle
$$
and 
$$
p_{\text{even}}(x)=\langle+| U_{\boldsymbol{\theta_{\text{even}}}}(x)|+\rangle
$$
then we can build $|0\rangle\langle0|\otimes U_{\boldsymbol{\theta_{\text{odd}}}}+|1\rangle\langle1|\otimes U_{\boldsymbol{\theta_{\text{even}}}}$ such that
\begin{align*}
\langle+|^{\otimes2} (|0\rangle\langle0|\otimes U_{\boldsymbol{\theta_{\text{odd}}}}+|1\rangle\langle1|\otimes U_{\boldsymbol{\theta_{\text{even}}}})(x)|+\rangle^{\otimes2}\\
=\frac{1}{2}\langle+| U_{\boldsymbol{\theta_{\text{odd}}}}(x)|+\rangle+\frac{1}{2}\langle+| U_{\boldsymbol{\theta_{\text{even}}}}(x)|+\rangle
=\frac{1}{2}p_{\text{odd}}(x)+\frac{1}{2}p_{\text{even}}(x)=p(x).
\end{align*}
\end{proof}

In fact, $\forall x \in [-1,1]$, the condition $|p(x)| \le \tfrac12$ is stronger than necessary to guarantee the existence of a quantum model. All that is truly required is
\begin{equation}
\label{eq:true_bound}
|p_{\text {odd}}(x)|_{\infty},|p_{\text {even}}(x)|_{\infty}\leq 1,
\end{equation}
where $p_{\text {even }}(x)=p(x)+p(-x)$ and $p_{\text {odd }}(x)=p(x)-p(-x)$. However for simplicity, we henceforth replace these two separate bounds by the simple condition $|p(x)| \le \tfrac12$. Then we can do the Hadamard test to estimate  $p(x)=\langle+|^{\otimes2} (|0\rangle\langle0|\otimes U_{\boldsymbol{\theta_{1}}}(x)+|1\rangle\langle1|\otimes U_{\boldsymbol{\theta_{2}}}(x))|+\rangle^{\otimes2}$ and analyze the circuit complexity. 

\begin{restate}{Proposition}{prop:univariat_poly}
For any real polynomial $p(x) \in \mathbb{R}[x]$ that satisfies $\operatorname{deg}(p(x)) \leq L$ and $\forall x \in[-1,1],|p(x)| \leq \frac{1}{2}$, there exists a quantum model $\mathcal{Q}$ that consists of a PQC $W_p(\boldsymbol{x})$ and an observable $Z^{(0)}$ such that 
$$
f_{\mathcal{Q}}(x):=\langle 0| W_p^{\dagger}(x) Z^{(0)} W_p(\boldsymbol{x})|0\rangle=p(x),
$$
where $Z^{(0)}$ is the Pauli $Z$ observable on the first qubit. The width of the PQC is at most $3$, the number of parameters is at most $2L+1$. The PQC $W_p(\boldsymbol{x})$ can be expressed as at most $4L$ double‐controlled rotation gates and $4$ Hadamard gates, with depth $4L+2$. Alternatively,  it can be expressed using $36L$ single‐qubit gates and $32L$ CNOT gates, with depth $60L-5$.
\end{restate}

The circuit diagram of $W_p(\boldsymbol{x})$ is shown in Figure \ref{fig:single_qubit_poly_main}.

\begin{proof}
We use Corollary \ref{coro:univariate_poly_coro} to build the circuit  
$$
W_p(\boldsymbol{x})=
\left(H\otimes H\otimes H\right)
\left(|0\rangle\langle0|\otimes U_{\boldsymbol{\theta_{1}}}(x)+|1\rangle\langle1|\otimes U_{\boldsymbol{\theta_{2}}}(x)\right)
\left(H\otimes I\otimes I\right)
$$
If we consider quantum computers that can implement multi-qubit controlled rotation gates such as neutral atom platforms, the depths of double controlled rotation gates $U_{\boldsymbol{\theta_{1}}},U_{\boldsymbol{\theta_{2}}}$ are $2(L-1)+1,2L+1$, respectively. Hence the PQC $W_p(\boldsymbol{x})$ consists of at most $4L$ 2-qubit controlled rotation gates and $4$ Hadamard gates, and the circuit depth is $4L+2$. The numbers of parameters of $U_{\boldsymbol{\theta_{1}}},U_{\boldsymbol{\theta_{2}}}$ are $(L-1)+1,L+1$ without counting $x$, respectively. So the total number of parameters is $2L+1$.

If we consider quantum computers that cannot implement multi-qubit controlled rotation gates natively, from \citet[Lemma 5.4]{Barenco_1995} we know that for a controlled-$R_p$ gate where $R_p$ is the rotation gate with Pauli matrix $p\in{\{z,y}\}$, controlled-$R_p$ can be implemented by $2$ $R_p$ gates and $2$ CNOT gates, and a controlled-$R_x$ gate can be implemented by a controlled-$R_z$ plus $2$ Hadamard gates. So a controlled-$R_z$ gate consists of $2$ single-qubit gates and $2$ CNOT gates with depth $4$, and a controlled-$R_x$ gate consists of $4$ single-qubit native gates and $2$ CNOT gates with depth $6$ because of $R_x(\theta)=H R_z(\theta) H$ where $H$ is Hadamard gates.

Then from \citet[Lemma 6.1]{Barenco_1995} we also know that a double controlled-$U$ gate can be implemented by $2$ controlled-$V$ (where $V^2=U$), $1$ controlled-$V^\dagger$ and $2$ CNOT gates, so a double controlled-$R_z$ can be implemented by $6$ single-qubit gates and $8$ CNOT gates with depth $12$, and a double controlled-$R_x$ can be implemented by $12$ single-qubit gates and $8$ CNOT gates with depth $18$. The depth arises from the specific circuit layout described in \citet[Lemma 6.1]{Barenco_1995}. 

A double-controlled-$ U_{\boldsymbol{\theta_{1}}}$ can be implemented by $L-1$ double controlled $R_x$ gates and $L$ double controlled $R_z$ gates, which means $12(L-1)+6L$ single-qubit gates and $8(L-1)+8L$ CNOT gates with depth $18(L-1)+12L=30L-18$. And a double-controlled-$ U_{\boldsymbol{\theta_{2}}}$ can be implemented by $L$ $R_x$ gates and $L+1$ $R_z$ gates, which means  $12L+6(L+1)$ single-qubit gates and $8L+8(L+1)$ CNOT gates with depth $18L+12(L+1)=30L+12$. 

Finally, with the extra $4$ Hadamard gates and $2$ Pauli-X gates (for controlled $|0\rangle$ to $|1\rangle$), we can conclude that the circuit requires $36L$ single-qubit gates and $32L$ CNOT gates with depth $60L-5$. Here, the extra depth due to Pauli-X gates can be ignored. It is worth noting that the bound we provide, derived from \citet{Barenco_1995}, is only an old upper bound. Improved results may exist, yielding tighter estimates for both gate count and circuit depth.
\end{proof}

In fact, Proposition \ref{prop:univariat_poly} is similar to the result in \citet[Theorem 56]{Gily_n_2019}, but we do not use their block-encoded method. Moreover, our result provides an extension to this existing work from two perspectives: clear quantum circuit implementation of univariate polynomials and exact quantification of all circuit resources. Actually, we can implement multi-qubit controlled rotation gates on the neutral atom platform quantum computer in an efficient and native way \citep{PhysRevA.103.062607}. We observe that building multi-qubit quantum gates directly in a “native” method is more efficient and valuable than constructing multi-qubit gates through a decomposition of two-qubit and single-qubit gates. 

Lemma \ref{lem:init_poly} and Corollary \ref{coro:univariate_poly_coro} give us quantum circuit expressivity analysis for univariate polynomials. Next we will talk about the multivariate polynomial cases.

\subsubsection{Multivariate case}
\label{sec:multivariate_uniform}
We now turn to discuss how to implement multivariate polynomials, and we define the $L^{\infty}$-norm of a polynomial with input variable $x\in E$ as $\|p\|_\infty:=\sup_{x\in E} |p(x)|$.
\begin{corollary}
\label{coro:constraint_mul_poly}
Given a multivariate polynomial $p(\boldsymbol{x})=\prod_{j=1}^D p_{j}(x_j)$ in $\boldsymbol{x}=(x_1,\dots,x_D) \in[-1,1]^D$ such that $\operatorname{deg}(p_j(x_j)) \leq L$, $p_j(x_j)$ has parity $L \bmod 2$, and $\|p_j\|_\infty \leq1$, there exists a PQC $U^{p}(\boldsymbol{x})$ such that
$$
\langle+|^{\otimes D} U^{p}(\boldsymbol{x})|+\rangle^{\otimes D}=p(\boldsymbol{x}).
$$
The width of the PQC is at most $D$, the depth is at most $L+1$, and the number of parameters is at most $LD$. 
\end{corollary}
The circuit diagram of $U^{p}(\boldsymbol{x})$ is shown in Figure \ref{fig:U^p_main}. 
\begin{proof}
From Lemma \ref{lem:init_poly}, there exist $D$ single-qubit PQCs $U_{\boldsymbol{\theta}_1}^{p_1}\left(x_1\right), U_{\boldsymbol{\theta}_2}^{p_2}\left(x_2\right), \ldots, U_{\boldsymbol{\theta}_D}^{p_D}\left(x_D\right)$ such that
\begin{equation}
\label{eq:wastful_expression}
\begin{aligned}
\langle+| U_{\boldsymbol{\theta}_1}^{n_1}\left(x_1\right)|+\rangle & =p_1(x_1),\\
\langle+| U_{\boldsymbol{\theta}_2}^{n_2}\left(x_2\right)|+\rangle & =p_2(x_2), \\
\cdots & \\
\langle+| U_{\boldsymbol{\theta}_D}^{n_D}\left(x_D\right)|+\rangle & =p_D(x_D),
\end{aligned}
\end{equation}
where each real polynomial $p_j(x_j) \in \mathbb{R}[x_j],\forall j\in[D]$ satisfies $\operatorname{deg}(p_j(x_j)) \leq L$, $p_j(x_j)$ has parity $L \bmod 2$, and $\|p_j\|_\infty \leq1$. Then we can define $p(\boldsymbol{x})=\prod_j^D p_j(x_j)$ and
\begin{equation}
\label{eq:monomial}
U^{p}(\boldsymbol{x})=\bigotimes_{j=1}^D U_{\boldsymbol{\theta}_j}^{p_j}\left(x_j\right),
\end{equation}
which gives
$$
\langle+|^{\otimes D} U^{p}(\boldsymbol{x})|+\rangle^{\otimes D}=\prod_{j=1}^D\langle+| U_{\boldsymbol{\theta}_j}^{p_j}\left(x_j\right)|+\rangle=p(\boldsymbol{x}).
$$
We can conclude that the depth of $U^{p}(\boldsymbol{x})$ is at most $L+1$, and the number of parameters is at most $LD$.
\end{proof}

Based on Corollary \ref{coro:constraint_mul_poly}, we can also consider the univariate constrained polynomial $p_j(x_j)$ as a monomial:

\begin{corollary}
\label{coro:monomial}
Given a monomial $ c_{\boldsymbol{n}} \boldsymbol{x}^{\boldsymbol{n}}=c_{\boldsymbol{n}} x_1^{n_1} x_2^{n_2} \cdots x_D^{n_D}$ in $\boldsymbol{x} \in[-1,1]^D$ such that $\left|c_{\boldsymbol{n}} \right| \leq 1$ and $\|\boldsymbol{n}\|_\infty \leq L$ for $\boldsymbol{n}=\left(n_1, \ldots, n_D\right) \in[L]^D$, there exists a PQC $U^{\boldsymbol{n}}(\boldsymbol{x})$ such that
$$
\langle+|^{\otimes D} U^{\boldsymbol{n}}(\boldsymbol{x})|+\rangle^{\otimes D}=c_{\boldsymbol{n}} \boldsymbol{x}^{\boldsymbol{n}}.
$$
The width of the PQC is at most $D$, the depth is at most $L+1$, and the number of parameters is at most $LD$. 
\end{corollary}
\begin{proof}
We can regard the univariate constrained polynomial $p_j(x_j)$ in Corollary \ref{coro:constraint_mul_poly} as a monomial, and the resource required would be the same.
\end{proof}

Corollary \ref{coro:monomial} is almost the same as \citet[Lemma S5]{yu2024nonasymptotic}, but we extend the domain of definition of $x$ from $[0,1]^D$ to $[-1,1]^D$, and adjust the condition from $\|\boldsymbol{n}\|_1 \leq L$ to $\|\boldsymbol{n}\|_\infty \leq L$. Since monomials can be expressed explicitly, any real multivariate polynomial can be written as a sum of monomials using the LCU technique. Hence, we have:

\begin{proposition}
\label{prop:multivariate_uniform}
For any real multivariate polynomial $p(\boldsymbol{x})=\sum_{\boldsymbol{n} } c_{\boldsymbol{n}} \boldsymbol{x}^{\boldsymbol{n}}$ with $D$ variables and degree at most $L$ in each variable such that $\|p\|_\infty\leq1$, $\boldsymbol{x} \in[-1,1]^D$, and the absolute values of all coefficients are smaller than $1/T$ where $T$ is the number of monomials of  $p(\boldsymbol{x})$, there exists a quantum model $\mathcal{Q}$ that consists of a PQC $W_p(\boldsymbol{x})$ and an observable $Z^{(0)}$ such that
$$
f_{\mathcal{Q}}(\boldsymbol{x}):=\langle 0| W_p^{\dagger}(\boldsymbol{x}) Z^{(0)} W_p(\boldsymbol{x})|0\rangle=p(\boldsymbol{x}),
$$
where $Z^{(0)}$ is the Pauli $Z$ observable on the first qubit. The width of the PQC is at most $D+\lceil\log T\rceil+1$, the depth $O(TDL \log T)$, and the number of parameters is at most $TD(L+1)$.
\end{proposition}

The circuit diagram of $W_p(\boldsymbol{x})$ is shown in Figure \ref{fig:theorem1}.
\begin{figure}[H]
    \centering
    \includegraphics[width=0.6\linewidth]{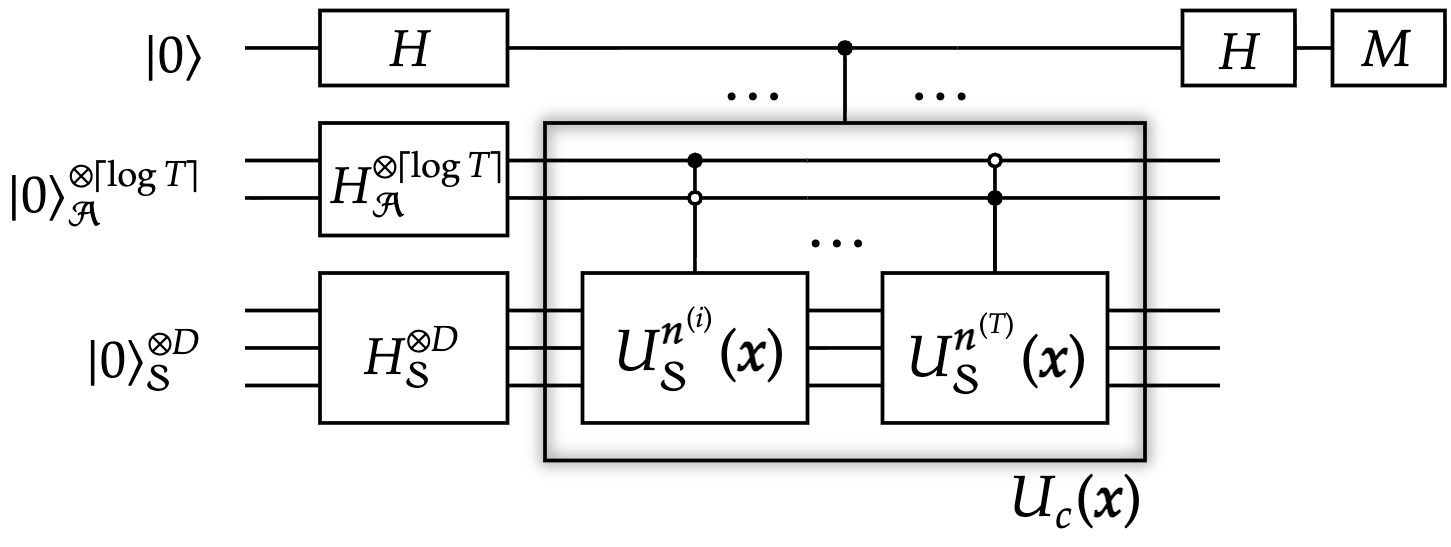}
    \caption{The circuit of $ W_p(\boldsymbol{x})$ consists of a single-qubit system, a $\lceil\log T\rceil$-qubits system $\mathcal{A}$ where $T$ is the number of monomials of  $p(\boldsymbol{x})$, and a $D$-qubits system $\mathcal{S}$. We can consider $ W_p(\boldsymbol{x})$ as the Hadamard test to estimate $p(\boldsymbol x)=\langle+|^{\otimes \lceil\log T\rceil}_{\mathcal{A}}
 \langle+|^{\otimes D}_{\mathcal{S}}
 U_c(\boldsymbol{x})
 |+\rangle^{\otimes \lceil\log T\rceil}_{\mathcal{A}}
 |+\rangle^{\otimes D}_{\mathcal{S}}$ such that $U_c(\boldsymbol{x})=\sum_{i=1}^T\ket{i}\bra{i}_\mathcal{A}\otimes U^{\boldsymbol{n}^{(i)}}_\mathcal{S}(\boldsymbol{x})$ where $U^{\boldsymbol{n}^{(i)}}_\mathcal{S}(\boldsymbol{x})$ is defined in Equation \eqref{eq:monomial} and shown in Figure \ref{fig:U^p_main}. }
    \label{fig:theorem1}
\end{figure}

\begin{proof}
Consider the given real multivariate polynomial $p(\boldsymbol{x})$ as
$$
p(\boldsymbol{x})=
\sum_{i=1}^T c_{\boldsymbol{n}^{(i)}}
   \,\boldsymbol{x}^{\boldsymbol{n}^{(i)}}
\quad\text{with}\;
\lvert c_{\boldsymbol n^{(i)}}\rvert\le\frac1T.
$$
where $\boldsymbol{x}=\left(x_1, \ldots, x_D\right) \in \mathbb{R}^D$, $\boldsymbol{n}=\left(n_1, \ldots, n_D\right) \in [L]^D$, $c_{\boldsymbol{n}} \in \mathbb{R}$, $\boldsymbol{x}^{\boldsymbol{n}}=x_1^{n_1} x_2^{n_2} \cdots x_D^{n_D}$, and $T$ denotes the number of monomials $c_{\boldsymbol{n}^{(i)}}
   \boldsymbol{x}^{\boldsymbol{n}^{(i)}}$, with $T\leq(L+1)^D$. We want to use a quantum circuit to express $p(\boldsymbol{x})$, so we define a controlled unitary $U_c(\boldsymbol{x})$ on the composite system $\mathcal{A}\otimes\mathcal{S}$ consisting of the $\lceil\log T\rceil$-qubit ancilla system $\mathcal{A}$ and the $D$-qubit system $\mathcal{S}$ such that
$$
U_c(\boldsymbol{x})=\sum_{i=1}^T\ket{i}\bra{i}_\mathcal{A}\otimes U^{\boldsymbol{n}^{(i)}}_\mathcal{S}(\boldsymbol{x})
$$
and use Corollary \ref{coro:monomial}:
$$
\langle+|^{\otimes D} U^{\boldsymbol{n}^{(i)}}_\mathcal{S}(\boldsymbol{x})|+\rangle^{\otimes D}=
C_{\boldsymbol{n}^{(i)}}\boldsymbol{x}^{\boldsymbol{n}^{(i)}}
$$
where we choose
\[
C_{\boldsymbol n^{(i)}}
\;=\;T\,c_{\boldsymbol n^{(i)}},
\bigl\lvert C_{\boldsymbol n^{(i)}}\bigr\rvert\le1, \quad
\text{since} \;\lvert c_{\boldsymbol n^{(i)}}\rvert\le\frac1T
\]
and
$$
H_\mathcal{A}^{\otimes \lceil\log T\rceil}|0\rangle^{\otimes \lceil\log T\rceil}_\mathcal{A}=|+\rangle^{\otimes \lceil\log T\rceil}_\mathcal{A}
$$
and we obtain
\begin{align*}
&\langle+|^{\otimes \lceil\log T\rceil}_{\mathcal{A}}
 \langle+|^{\otimes D}_{\mathcal{S}}
 U_c(\boldsymbol{x})
 |+\rangle^{\otimes \lceil\log T\rceil}_{\mathcal{A}}
 |+\rangle^{\otimes D}_{\mathcal{S}}\\
 &\quad
 =\langle+|^{\otimes \lceil\log T\rceil}_\mathcal{A} \langle+|^{\otimes D}_\mathcal{S}
\left(\sum_{i=1}^T\ket{i}\bra{i}_\mathcal{A}\otimes U^{\boldsymbol{n}^{(i)}}_\mathcal{S}(\boldsymbol{x})\right)
|+\rangle^{\otimes \lceil\log T\rceil}_\mathcal{A} |+\rangle^{\otimes D}_\mathcal{S}\\
&\quad
= \sum_{i=1}^T
  \langle+|^{\otimes \lceil\log T\rceil}_{\mathcal{A}}\ket{i}_{\mathcal{A}}
  \bra{i}_{\mathcal{A}}
  |+\rangle^{\otimes \lceil\log T\rceil}_{\mathcal{A}}
  \langle+|^{\otimes D}_{\mathcal{S}}
  U^{\boldsymbol{n}^{(i)}}_{\mathcal{S}}(\boldsymbol{x})
  |+\rangle^{\otimes D}_{\mathcal{S}}\\
&\quad
= \sum_{i=1}^T
  \Bigl(\tfrac1{\sqrt{T}}\Bigr)
  \Bigl(\tfrac1{\sqrt{T}}\Bigr)
  \;C_{\boldsymbol{n}^{(i)}}\,
   \boldsymbol{x}^{\boldsymbol{n}^{(i)}}\\
&\quad
= \sum_{i=1}^T \frac{1}{T}C_{\boldsymbol{n}^{(i)}}
   \,\boldsymbol{x}^{\boldsymbol{n}^{(i)}}
   =\sum_{i=1}^T c_{\boldsymbol n^{(i)}}\,\boldsymbol x^{\boldsymbol n^{(i)}}
=p(\boldsymbol x)
\end{align*}
This explains why we set $\lvert c_{\boldsymbol n^{(i)}}\rvert\le\frac1T$, to satisfy $\bigl\lvert C_{\boldsymbol n^{(i)}}\bigr\rvert\le1$ so that we can use Corollary \ref{coro:monomial} to express all monomials and the polynomial. Then we still use the Hadamard test to estimate 
$$
p(\boldsymbol x)=\langle+|^{\otimes \lceil\log T\rceil}_{\mathcal{A}}
 \langle+|^{\otimes D}_{\mathcal{S}}
 U_c(\boldsymbol{x})
 |+\rangle^{\otimes \lceil\log T\rceil}_{\mathcal{A}}
 |+\rangle^{\otimes D}_{\mathcal{S}}
 $$
The controlled unitary $ U_c(\boldsymbol{x})=\sum_{i=1}^T\ket{i}\bra{i}_\mathcal{A}\otimes U^{\boldsymbol{n}^{(i)}}_\mathcal{S}(\boldsymbol{x})$ can be implemented by $T$ occurrences of $U^{\boldsymbol{n}^{(i)}}_\mathcal{S}(\boldsymbol{x})$ controlled by $\lceil\log T\rceil$ qubits. From Equation \eqref{eq:monomial}, each $U^{\boldsymbol{n}^{(i)}}_\mathcal{S}(\boldsymbol{x})=\bigotimes_{j=1}^D U_{\boldsymbol{\theta}_j}^{n_j}\left(x_j\right)$ can be implemented by $D$ occurrences of $U_{\boldsymbol{\theta}_j}^{n_j}\left(x_j\right)$ defined in Equation \eqref{eq:u_theta}, which consists of $O(L)$ occurrences of single-qubit rotation gates. In total, a $U_c(\boldsymbol{x})$ can be implemented by $O(TDL)$ single-qubit rotation gates controlled by $\lceil\log T\rceil$ qubits.

We know that a single-qubit rotation gate controlled by $\lceil\log T\rceil$ qubits could be implemented by a quantum circuit consisting of CNOT gates and single-qubit gates with depth $O(\log T)$ \citep{da_Silva_2022}. Thus $U_c(\boldsymbol{x})$ could be implemented by a quantum circuit with depth $O(TDL \log T)$ and width $D+\lceil\log T\rceil$. 

Then we build the Hadamard test circuit $W_p^{\dagger}(\boldsymbol{x})$ to estimate $p(\boldsymbol x)$ such that
\begin{equation}
\label{eq:theorem1}
W_p(\boldsymbol{x})=
(H\otimes H^{D+\lceil\log T\rceil})
\left(|0\rangle\langle0|\otimes I+|1\rangle\langle1|\otimes  U_c(\boldsymbol{x})\right)
\left(H\otimes I\otimes I\right)
\end{equation}
The extra controlled-$U_c(\boldsymbol{x})$ and Hadamard gates do not change the depth complexity, but only increase the circuit width by $1$. Hence we have the circuit $W_p^{\dagger}(\boldsymbol{x})$ with width $D+\lceil\log T\rceil+1$, depth $O(TDL \log T)$, and the number of parameters $TD(L+1)$. 
\end{proof}

As $\boldsymbol{n}=\left(n_1, \ldots, n_D\right) \in [L]^D$, the number of monomials $T\leq(L+1)^D$, we can adjust the Proposition \ref{prop:multivariate_uniform} to eliminate the parameter $T$ so that the bound for coefficients absolute values becomes $1/(L+1)^D$. But keeping the coefficients so small obstructs optimization, and the maximum absolute value of circuit output is only 1. To solve this we can simply multiply the final output of the circuit by a factor $\Lambda$ to express more general polynomials while keeping the quantum model ($W_p(\boldsymbol{x})$ and $Z^{(0)}$) unchanged.

\begin{restate}{Theorem}{theorem:multivariate_uniform}
Let  $p(\boldsymbol{x})=\sum_{\boldsymbol{n} } c_{\boldsymbol{n}} \boldsymbol{x}^{\boldsymbol{n}}$ be any real multivariate polynomial with $D$ variables and degree at most $L$ in each variable such that $\boldsymbol{x} \in[-1,1]^D$, $c_{\boldsymbol{n}}\in \mathbb R$, $\boldsymbol{n} \in[L]^D$. Then there exists a quantum model $\mathcal{Q}$ that consists of a PQC $W_p(\boldsymbol{x})$ and an observable $Z^{(0)}$ with the scaling factor $\Lambda$ such that
$$
f_{\mathcal{Q}}(\boldsymbol{x}):=\langle 0| W_p^{\dagger}(\boldsymbol{x}) Z^{(0)} W_p(\boldsymbol{x})|0\rangle=p(\boldsymbol{x})/\Lambda,
$$
where $\Lambda=|c_{\boldsymbol{n}}|_\infty  (L+1)^D$ and $Z^{(0)}$ is the Pauli $Z$ observable on the first qubit. The width of the PQC is $O(D \log L)$, the depth is $O\left(D^2 L^{D+1} \log L\right)$, and the number of parameters is $O\left(D L^{D+1}\right)$.
\end{restate}

The circuit diagram of $W_p(\boldsymbol{x})$ is shown in Figure \ref{fig:Circuit_diagram_withoutT_main}.

\begin{proof}
Following Proposition \ref{prop:multivariate_uniform}, replace $T$ by $(L+1)^D$, we then obtain the width is at most $D(\log (L+1)+1)+2$, the depth is $O\left((L+1)^{D+1}D^2\log (L+1)\right)$, and the number of parameters is at most $(L+1)^DD((L+1)+1)+D$. Using the complexity notation, we have the width is $O(D \log L)$, the depth is $O(D^2 L^{D+1} \log L)$, and the number of parameters is at most $O(D L^{D+1})$.

Then we need to satisfy the bound on absolute coefficient values and the property that the quantum circuit measurements can only output values with absolute value at most $1$. We can divide the polynomial by a scaling factor $\Lambda$ to satisfy the two constraints together. Precisely, for any $p(\boldsymbol{x})=\sum_{\boldsymbol{n} } c_{\boldsymbol{n}} \boldsymbol{x}^{\boldsymbol{n}}$, the maximum absolute coefficient value of $p(\boldsymbol{x})/\Lambda$ is $|c_{\boldsymbol{n}}|_\infty/\Lambda$ such that
$$
|c_{\boldsymbol{n}}|_\infty/\Lambda \leq \frac{1}{(L+1)^D},
$$
the smallest scaling factor is
$$
\Lambda=|c_{\boldsymbol{n}}|_\infty  (L+1)^D.
$$
Once the bound condition on absolute coefficient values is satisfied, we also have $|p(\boldsymbol{x})/\Lambda|_\infty \leq1$. 

\end{proof}

Theorem  \ref{theorem:multivariate_uniform}  is similar to the results of  \citet[Theorem 1]{yu2024nonasymptotic}: For any multivariate polynomial $p(\boldsymbol{x})=\sum_{\boldsymbol{n} } c_{\boldsymbol{n}} \boldsymbol{x}^{\boldsymbol{n}}$ with $D$ variables and $\|\boldsymbol{n}\|_1\leq L$ such that $|p(\boldsymbol{x})| \leq$ 1 for $\boldsymbol{x} \in[0,1]^D$, there exists a PQC $W_p(\boldsymbol{x})$ such that $\langle 0| W_p^{\dagger}(\boldsymbol{x}) Z^{(0)} W_p(\boldsymbol{x})|0\rangle=p(\boldsymbol{x})$, with PQC width $O(D+\log L+L \log D)$, the depth is $O\left(L^2 D^L(\log L+L \log D)\right)$, and the number of parameters is $O\left(L D^L(L+D)\right)$. 

In fact, they consider a stricter condition $\|\boldsymbol{n}\|_1\leq L$ instead of our $\|\boldsymbol{n}\|_\infty\leq L$, use the bound $T \leq(L+1) D^L$ instead of our $T=(L+1)^D$, and take the domain $\boldsymbol{x} \in[0,1]^D$ instead of our  $\boldsymbol{x} \in[-1,1]^D$. Moreover, they would also require a scaling factor to ensure that the PQC is found for all monomials of eligible polynomials. Compared to their results, our condition not only specifies the scaling factor bound but also has fewer restrictions on the degree and input domain. Moreover, our complexity result is clearly better when $L$ is larger than $D$, which is also the more common case in certain complex financial problems.

\subsection{Tensor-decomposed polynomial implementation}
\label{sec_app: TD poly}

\begin{restate}{Theorem}{theorem:tensor_poly}
Let $p(\boldsymbol{x})$ be any real multivariate polynomial with $D$ variables and degree at most $L$ in each variable such that 
$$
p(\boldsymbol{x})=\sum_{i=1}^R \lambda_i \prod_{j=1}^D p_{i, j}\left(x_j\right),
$$
where $R\in[1,(L+1)^D]$, $\sum_{i=1}^R |\lambda_i|=\Lambda\in\mathbb{R}$, and each $p_{i,j}(x_j)$ is a univariate polynomial of degree $L$ satisfying $|p_{i,j}(x_j)| \leq 1/2$  for $x_j \in[-1,1],\forall i\in[R],j\in[D]$. Then there exists a quantum model $\mathcal{Q}$ that consists of a PQC $W_p(\boldsymbol{x})$ and an observable $Z^{(0)}$ such that
$$
f_{\mathcal{Q}}(\boldsymbol{x}):=\langle 0| W_p^{\dagger}(\boldsymbol{x}) Z^{(0)} W_p(\boldsymbol{x})|0\rangle=p(\boldsymbol{x})/\Lambda,
$$
where $Z^{(0)}$ is the Pauli $Z$ observable on the first qubit. The width of the PQC is at most $2D+\lceil\log R\rceil+1$, the depth is $O(RDL \log R)$, and the number of parameters is $O\left(RDL\right)$.
\end{restate}

The circuit diagram of the PQC $W_p(\boldsymbol{x})=
(H\otimes F_\mathcal{A} \otimes H_\mathcal{S}^{\otimes2D})
\left(|0\rangle\langle0|\otimes I+|1\rangle\langle1|\otimes  V_c(\boldsymbol{x})\right)
\left(H\otimes I_\mathcal{A}\otimes I_\mathcal{S}\right)$, with $V_c(\boldsymbol{x})=\sum_{i=1}^R \ket{i}\bra{i}_\mathcal{A}\;\otimes\;V^{(i)}_\mathcal{S}(\boldsymbol{x})$, and $V^{(i)}_\mathcal{S}(\boldsymbol{x})$ are shown in Figure \ref{fig:theorem2_circuit_main} and  Figure \ref{fig:V_circuit_main}, respectively.

\begin{proof}
For each term \(i\in[R]\) and each variable \(x_j\), apply Corollary \ref{coro:univariate_poly_coro} to the univariate polynomial \(p_{i,j}(x_j)\).  Since \(\deg(p_{i,j})\le L\) and \(\lvert p_{i,j}(x_j)\rvert\le\tfrac12\), there exist angles \(\boldsymbol{\theta}_{i,j}^{(1)}\in\mathbb{R}^L\) and \(\boldsymbol{\theta}_{i,j}^{(2)}\in\mathbb{R}^{L+1}\) such that
\[
p_{i,j}(x_j)
=\bigl\langle +\bigr|^{\otimes 2}\!
\Bigl(\ket{0}\!\bra{0}\otimes U_{\boldsymbol{\theta}_{i,j}^{(1)}}(x_j)
+\ket{1}\!\bra{1}\otimes U_{\boldsymbol{\theta}_{i,j}^{(2)}}(x_j)\Bigr)
\bigl|+\bigr\rangle^{\otimes 2}.
\]
Define a \(2D\)-qubit unitary for each \(i\) by
\[
V^{(i)}(\boldsymbol{x})
=\bigotimes_{j=1}^D
\Bigl(\ket{0}\!\bra{0}\otimes U_{\boldsymbol{\theta}_{i,j}^{(1)}}(x_j)
+\ket{1}\!\bra{1}\otimes U_{\boldsymbol{\theta}_{i,j}^{(2)}}(x_j)\Bigr),
\]
and consider \(V^{(i)}\) in the \(2D\) qubit system \(\mathcal{S}\) such that,
\begin{equation}
\label{eq:V^i}
\bigl\langle +\bigr|^{\otimes 2D}_\mathcal{S}
V^{(i)}_\mathcal{S}(\boldsymbol{x})\,
\bigl|+\bigr\rangle^{\otimes 2D}_\mathcal{S}
=\prod_{j=1}^D p_{i,j}(x_j).
\end{equation}
Define $F_\mathcal{A}$ on the ancilla register \(\mathcal{A}\) of \(\lceil\log R\rceil\) qubits such that
\[
F_\mathcal{A}\,\ket{0}^{^{\otimes\lceil\log R\rceil}}_\mathcal{A}
=\sum_{i=1}^R\sqrt{w_i}\,|i\rangle_\mathcal{A}
\quad \text{where} \;
w_i \;=\;\frac{\lvert \lambda_i\rvert}{\sum_{k=1}^R \lvert \lambda_k\rvert}
\]
with $\sum_{i=1}^R w_i=1$, and define the controlled unitary $V_c(\boldsymbol{x})$ such that
\[
V_c(\boldsymbol{x})
=\sum_{i=1}^R \ket{i}\bra{i}_\mathcal{A}\;\otimes\;V^{(i)}_\mathcal{S}(\boldsymbol{x}).
\]
Define $|v_p\rangle=\sum_{i=1}^R\sqrt{w_i}\,|i\rangle_\mathcal{A}\otimes |+\rangle_{\mathcal{S}}^{\otimes2D}$, we have
\begin{align*}
&\langle v_p|\,V_c(\boldsymbol{x})\,|v_p\rangle\\
&=
\Bigl(\sum_{i'=1}^R\sqrt{w_{i'}}\,\langle i'|_\mathcal{A}\otimes \langle+|_{\mathcal{S}}^{\otimes2D}\Bigr)
\Bigl(\sum_{i=1}^R\ket{i}\bra{i}_\mathcal{A}\otimes V^{(i)}_\mathcal{S}(\boldsymbol{x})\Bigr)
\Bigl(\sum_{k=1}^R\sqrt{w_k}\,\ket{k}_\mathcal{A}\otimes\ket{+}_{\mathcal{S}}^{\otimes2D}\Bigr)\\
&=\sum_{i',i,k}\sqrt{w_{i'}w_k}\,
\langle i'|i\rangle_\mathcal{A}\langle i|k\rangle_\mathcal{A}\,
\langle+|V^{(i)}_\mathcal{S}(\boldsymbol{x})|+\rangle_{\mathcal{S}}^{\otimes2D}\,
=\sum_{i=1}^R w_i\,
\bigl\langle +\bigr|_{\mathcal{S}}^{\otimes2D}\,
V^{(i)}_\mathcal{S}(\boldsymbol{x})\,
\bigl|+\bigr\rangle_{\mathcal{S}}^{\otimes2D}\\
&=\sum_{i=1}^R w_i\prod_{j=1}^D p_{i,j}(x_j)
\;\\
&=\;\frac{1}{\sum_{k=1}^R|\lambda_k|}\sum_{i=1}^R|\lambda_i|\prod_{j=1}^D p_{i,j}(x_j)
\;\\
&\equiv\;\frac{1}{\sum_{k=1}^R|\lambda_k|}\sum_{i=1}^R\lambda_i\prod_{j=1}^D p_{i,j}(x_j)
\\&=\;p(\boldsymbol{x})/\Lambda.
\end{align*}
In the second-to-last row, we can assume $\sum_{i=1}^R|\lambda_i|\prod_{j=1}^D p_{i,j}(x_j)
\;\equiv\;\sum_{i=1}^R\lambda_i\prod_{j=1}^D p_{i,j}(x_j)$ without loss of generality, since the sign of $\lambda_i$ can always be absorbed into one of the univariate polynomials $p_{i, j}\left(x_j\right)$. Then we can do the Hadamard test to estimate $\langle v_p|\,V_c(\boldsymbol{x})\,|v_p\rangle=p(\boldsymbol{x})/\Lambda$.

Using similar ideas as in Proposition \ref{prop:multivariate_uniform}, $V_c(\boldsymbol{x})$ could be implemented by a quantum circuit with depth $O(RDL \log R)$ and width $2D+\lceil\log R\rceil$. For the Hadamard test part, \citet{mottonen2004transformationquantumstatesusing} showed that $F$ can be implemented by a quantum circuit of CNOT gates and single-qubit rotation gates of size $O(R)$, so the depth and width of the circuit are still the same. The extra controlled-$V_c(\boldsymbol{x})$ gates introduced by the Hadamard test also do not change the complexity. Note that the number of parameters in the PQC equals the number of parameters in $V_c(\boldsymbol{x})$ plus the number of parameters in $F$, i.e., $O(RDL)$ and $O(R)$, which is still $O(RDL)$ in total.
\end{proof}

\subsubsection{Rank-1 tensor-decomposed polynomial implementation}
\label{section:rank1_poly}
In this work, we specifically consider a simple case when $R=1$, i.e. rank-1 multivariate polynomial $p(\boldsymbol{x})$:
\begin{restate}{Corollary}{coro:rank1_poly}
Let  $p(\boldsymbol{x})=\prod_{j=1}^D p_{ j}\left(x_j\right)$ be any real multivariate polynomial such that $\forall j\in[D],x_j\in[-1,1]$, each $p_{i,j}(x_j)$ is a univariate polynomial of degree $L$ satisfying $|p_{i,j}(x_j)| \leq 1/2$. Then there exists a quantum model $\mathcal{Q}$ that consists of a PQC $W_p(\boldsymbol{x})$ and an observable $Z^{(0)}$ such that
$$
f_{\mathcal{Q}}(\boldsymbol{x}):=\langle 0| W_{p}^{\dagger}(\boldsymbol{x}) Z^{(0)}W_{p}(\boldsymbol{x})|0\rangle=p(\boldsymbol{x}),
$$
where $Z^{(0)}$ is the Pauli $Z$ observable on the first qubit. The width of the PQC is at most $2D+1$, the depth is at most $4LD+2$ by allowing double-controlled rotation gates to be implemented natively, and the number of parameters is at most $(2L+1)D$.
\end{restate}

\begin{proof}

For each variable, the depth of the double-controlled rotation gates is at most $2 L+1+2(L-1)+1=4 L$ and the number of parameters is at most $2 L+1$. Thus, with $D$ variables and including the additional Hadamard layers, the total circuit depth is $4 L D+2$, whereas the total number of parameters is at most $(2 L+1)D$.
\end{proof}

\section{QPINN Circuit for Volatility Input}
\label{sec:volatility_input_qpinn_circuit}

Figure~\ref{fig:qsp_circuit_volatility_input} shows the eight-qubit QPINN circuit used in Section~\ref{sec:parametric_volatility_experiment}. The circuit is wrapped into two rows only for readability, and the lower row continues the upper row from left to right.

\begin{figure}[p]
    \centering
    \includegraphics[angle=90,origin=c,width=0.47\textheight]{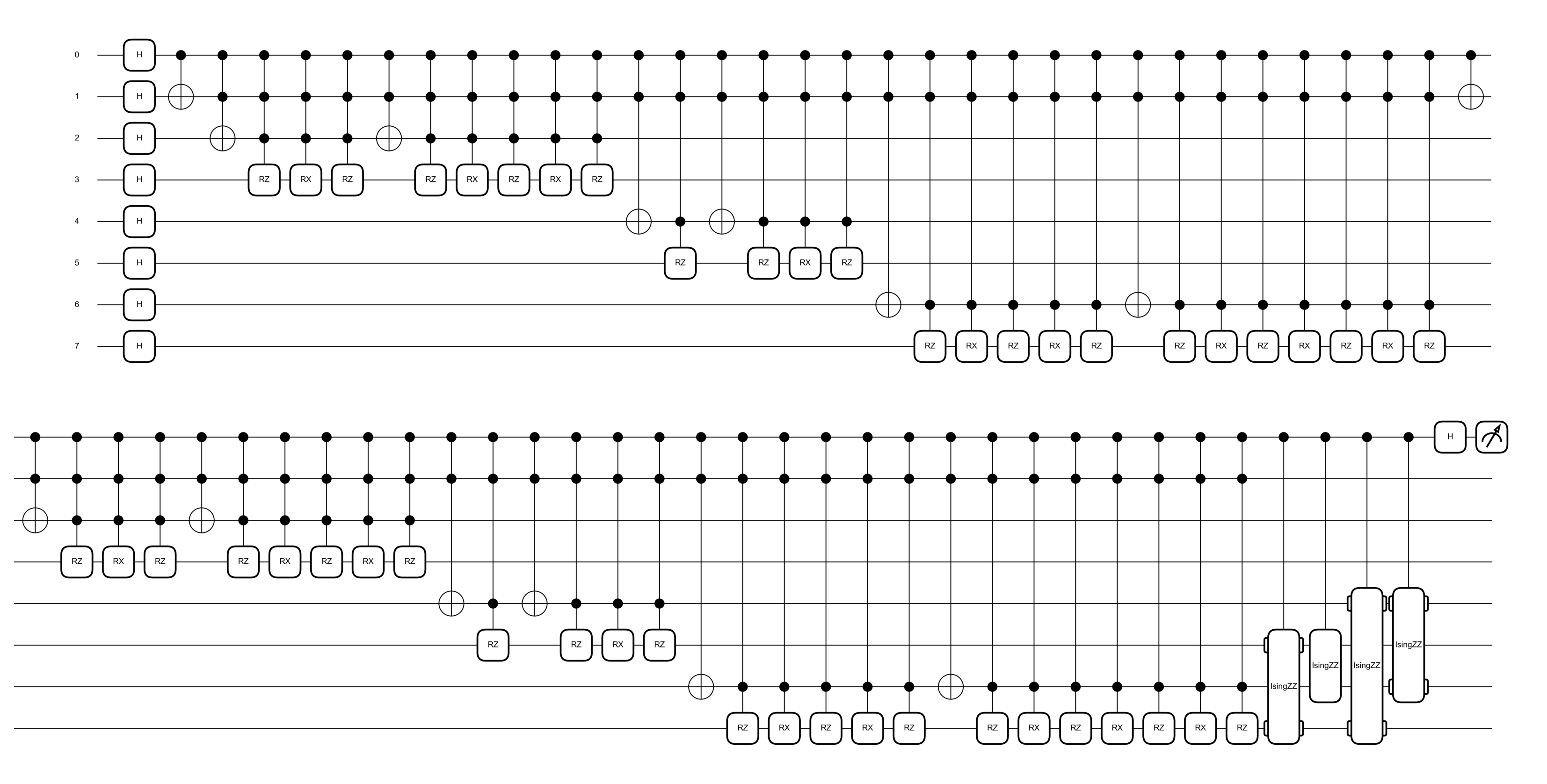}
    \caption{Complete circuit for the QPINN used in the volatility-input HJB PDE experiment. The first qubit is the Hadamard-test qubit and the second qubit is the tensor-rank selection qubit, and the remaining six qubits encode $(x,t,\sigma)$. The four \texttt{IsingZZ} gates form the added entangling layer between the $t$ and $\sigma$ registers.}
    \label{fig:qsp_circuit_volatility_input}
\end{figure}

The qubits are ordered from top to bottom as $q_0,\ldots,q_7$. The first qubit $q_0$ is the Hadamard-test qubit, and the second qubit $q_1$ selects the two tensor-rank components with $R=2$. The remaining six qubits form three two-qubit variable registers: $(q_2,q_3)$ encodes $x$, $(q_4,q_5)$ encodes $t$, and $(q_6,q_7)$ encodes $\sigma$. The tensor-decomposed model uses different univariate degrees for the three variables: $L_t=1$, $L_x=2$, and $L_\sigma=3$. The four \texttt{IsingZZ} gates at the right end form the added entangling layer. They act on $(q_5,q_7)$, $(q_5,q_6)$, $(q_4,q_7)$, and $(q_4,q_6)$, realizing the pairwise $ZZ$ couplings between the $t$ and $\sigma$ registers in Equation~\eqref{eq:volatility_input_added_hamiltonian}. The $x$ register is not directly entangled by this layer.




\end{appendices}



\end{document}